\documentclass[11pt]{article}
\pdfoutput=1

\usepackage{style}
\usepackage{shortcuts}

\title{Sparse Nonnegative Convolution Is Equivalent to\texorpdfstring{\\}{ }Dense Nonnegative Convolution\texorpdfstring{\thanks{This work is part of the project TIPEA that has received funding from the European Research Council (ERC) under the European Unions Horizon 2020 research and innovation programme (grant agreement No.~850979).}
}{}}
\author{Karl Bringmann}
\author{Nick Fischer}
\author{Vasileios Nakos}
\affil{Saarland University and Max Planck Institute for Informatics, Saarland Informatics Campus}

\begin{document}

\maketitle

\begin{abstract}
Computing the convolution $A \star B$ of two length-$n$ vectors $A, B$ is an ubiquitous computational primitive, with applications in a variety of disciplines. 
Within theoretical computer science, applications range from string problems to Knapsack-type problems, and from 3SUM to All-Pairs Shortest Paths. These applications often come in the form of \emph{nonnegative} convolution, where the entries of $A,B$ are nonnegative integers. The classical algorithm to compute $A\star B$ uses the Fast Fourier Transform (FFT) and runs in time $O(n \log n)$.

However, in many cases $A$ and $B$ might satisfy sparsity conditions, and hence one could hope for significant gains compared to the standard FFT algorithm. The ideal goal would be an $O(k \log k)$-time algorithm, where $k$ is the number of non-zero elements in the output, i.e., the size of the support of $A \star B$. This problem is referred to as \emph{sparse} nonnegative convolution, and has received a considerable amount of attention in the literature; the fastest algorithms to date run in time $O(k \log^2 n)$.

The main result of this paper is the first $O(k \log k)$-time algorithm for sparse nonnegative convolution. Our algorithm is randomized and assumes that the length $n$ and the largest entry of $A$ and $B$ are subexponential in $k$. Surprisingly, we can phrase our algorithm as a reduction from the sparse case to the dense case of nonnegative convolution, showing that, under some mild assumptions, sparse nonnegative convolution is equivalent to dense nonnegative convolution for constant-error randomized algorithms. Specifically, if $D(n)$ is the time to convolve two nonnegative length-$n$ vectors with success probability $2/3$, and $S(k)$ is the time to convolve two nonnegative vectors with output size $k$ with success probability $2/3$, then $S(k) = O( D(k) + k (\log \log k)^2)$.

Our approach uses a variety of new techniques in combination with some old machinery from linear sketching and structured linear algebra, as well as new insights on linear hashing, the most classical hash function.
\end{abstract}


\section{Introduction} \label{sec:introduction}
Computing convolutions is an ubiquitous task across all science and engineering. Some of its special cases are as fundamental as the general case; we first introduce the most important problem variants.

\begin{itemize}
\item \textbf{\textsf{Boolean Convolution}} is the problem of computing for given vectors $A,B \in \{0,1\}^n$ the vector $C = A \ostar B \in \{0,1\}^{2n-1}$ defined by $C_k = \bigvee_i A_i \wedge B_{k-i}$. This formalizes a situation in which we split a computational problem into two subproblems, so that in total there is a solution of size $k$ if and only if for some $i$ there is a solution of the left subproblem of size $i$ and there is a solution of the right subproblem of size $k - i$. Therefore, it is a natural task that frequently arises in algorithm design. Boolean convolution is also equivalent to \emph{sumset computation}, where for given sets $A,B \subseteq \{0,1,\ldots,n-1\}$ the task is to compute their sumset $A+B$ consisting of all sums $a+b$ with $a \in A,\, b \in B$. It therefore frequently comes up in algorithms for Subset Sum, 3SUM, and similar problems, see, e.g.,~\cite{ChanL15,Bringmann17,BateniHSS18,JansenR19,KoiliarisX19,MuchaWW19}. 

\item \textbf{\textsf{Nonnegative Convolution}} is the problem of computing for given vectors $A,B \in \Int^n$ with nonnegative entries the vector $C = A \conv B \in \Int^{2n-1}$ defined by $C_k = \sum_i A_i \cdot B_{k-i}$. For instance, if~$A_i$ and $B_i$ count the number of size-$i$ solutions of the left and right subproblem, then $C_k$ counts the number of size-$k$ solutions of the whole problem. It also comes up in string algorithms when computing the Hamming distance of a pattern and each sliding window of a text; this connection was found by Fischer and Paterson~\cite{FischerP74} and has been exploited in many string algorithms, see, e.g.,~\cite{Abrahamson87,AmirLP04,Karloff93,KopelowitzP18}. 
As an operation, nonnegative convolution is frequent also in computer vision, image processing and computer graphics; a prototypical such an example is the process of blurring an image by a Gaussian kernel in order to remove noise and detail~\cite{GaussianSmoothing}. Note that nonnegative convolution generalizes Boolean convolution, as $A \ostar B$ is simply the support of the nonnegative convolution $A \conv B$. In this paper our focus is on the nonnegative convolution problem.

\item \textbf{\textsf{General Convolution}}, or simply ``convolution'', denotes the general case obtained by dropping the nonnegativity assumption from the previous problem variant. This problem is central in signal processing and is also equivalent to polynomial multiplication, one of the most fundamental problems of computer algebra, and thus has a wealth of applications. 
We remark that general convolution can be reduced to nonnegative convolution (and thus they are equivalent), by replacing $A'_i := A_i + M$ and $B'_j := B_j + M$, which are nonnegative for sufficiently large $M$, and noting that $ (A \conv B)_k=(A' \conv B')_k~\mathrm{mod}~M$. However, this reduction destroys the \emph{sparsity} of the input and output, and thus is not applicable in the context of this paper.
\end{itemize}

\paragraph{Algorithms in the Dense Case.}
The standard algorithm for these problems uses Fast Fourier Transform (FFT) and runs in time $\Order(n \log n)$ on the RAM model. This running time is conjectured to be optimal (at least for general convolution), but proving this is a big open problem. There is some evidence in favor of this conjecture, for instance nonnegative convolution can be used to multiply integers and the latter is connected to the network coding conjecture~\cite{AfshaniFKL19}. For Boolean convolution, the evidence is less clear, since there exists a Boolean convolution algorithm by Indyk~\cite{Indyk98} running in time $O(n)$ with the guarantee that any fixed output entry is correct with constant probability (see \autoref{thm:indyk}). However, in this paper we focus on algorithms where the whole output vector is correct with constant probability, and boosting Indyk's algorithm to such a guarantee would again result in running time $O(n \log n)$. Therefore, for all three problem variants it is plausible that time $\Order(n \log n)$ is optimal, even for constant-error randomized algorithms.

\paragraph{Algorithms in the Sparse Case.}

A long line of work has considered convolution in a sparse setting, see, e.g.,~\cite{Muthukrishnan95,ColeH02,Roche08,MonaganP09,VanDerHoevenL12,ArnoldR15,ChanL15,Roche18,Nakos20,GiorgiGC20}. Here the running time is expressed not only in terms on $n$, but also in terms of the output size $k$, defined as the number of nonzero entries of $A \conv B$. All variants of convolution listed above admit randomized algorithms running in near-linear time $k \polylog n$. This was first achieved by Cole and Hariharan~\cite{ColeH02} for nonnegative convolution with a Las Vegas algorithm running in time $O(k \log^2 n + \polylog(n))$, in~\cite{Nakos20} for general convolution with running time\footnote{Here we use the notation $\widetilde\Order(T) = \bigcup_{c>0} O(T \log^c T)$.} $\widetilde\Order(k \log^2 n + \polylog(n))$. The latter was improved by Giorgi, Grenet and Perret du Cray~\cite{GiorgiGC20} to a \emph{bit complexity} of $\widetilde{O}(k \log n)$; it seems that on the RAM model their algorithm would run in time $O(k \log^5 k \polyloglog n)$.\footnote{To determine their running time on the RAM model, from the last paragraph of the proof of their Lemma~4.7 one can infer that the bottleneck of their running time stems from $\Theta(\log^2 k)$ many dense convolutions on vectors of length $\Theta(k \log(k \log n) \log(k \log\log n)) = \Theta(k \log^2 k \polyloglog n)$. Since one dense convolution of length $d$ can be performed in time $O(d \log d)$ on the RAM model, they require time $O(k \log^5 k \polyloglog n)$.} Implementations of sparse convolution algorithms exist in Maple~\cite{MonaganP14,MonaganP15} and Magma~\cite{Magma}.

This research is closely related to the extensively studied sparse Fourier transform problem, e.g.~\cite{GilbertGIMS02,GilbertMS05,HassaniehIKP12,IndykK14,IndykKP14}. Indeed, the same running time of $O(k \log^2 n)$, albeit with a more complicated algorithm and under the assumption that complex exponentials can be evaluated at constant time, can be obtained by combining the state-of-the-art sparse Fourier transform with the semi-equispaced Fourier transform, see \autoref{sec:prevtechniques}.

In summary, for nonnegative convolution on the RAM model, the state of the art requires time $\Omega(k \log^2 n)$ or $\Omega(k \log^5 k)$. In view of the conjecture that $O(n \log n)$ is optimal for the dense case, the best running time we could expect for the sparse case would be $O(k \log k)$. The driving question of this work is thus:

\begin{center}
  \emph{Can sparse nonnegative convolution be solved in time $O(k \log k)$?}
\end{center}

We note that the need for sparse convolution arises in many different areas of algorithm design, for example algorithms for the sparse cases of Boolean and nonnegative convolution have been used for clustered 3SUM and similar problems~\cite{ChanL15}, output-sensitive Subset Sum algorithms~\cite{BringmannN20}, pattern matching on point sets~\cite{CardozeS98}, sparse wildcard matching~\cite{ColeH02}, and other string problems~\cite{AmirBP14,AmirKP07}.

\subsection{Results} \label{sec:results}
We present a novel connection between the sparse and dense case of nonnegative convolution, which can be viewed as work at the intersection of sparse recovery and fine-grained complexity.

We work on the Word RAM model where each cell stores a word consisting of $w$ bits, and standard operations on $w$-bit integers can be performed in constant time; this includes addition, multiplication, and division with remainder. 
We always assume that the length $n$ of the input vectors as well as each input entry fit into a word, or more precisely into a constant number of words. For nonnegative convolution, this means that the input consists of vectors $A,B \in \{0,1,\ldots,\Delta\}^n$ with $n, \Delta \le 2^{O(w)}$.
In this machine model, the standard algorithm for dense convolution uses FFT and runs in time $D(n) = O(n \log n)$. In the following, we denote by $D_\delta(n)$ the running time of a randomized algorithm for dense nonnegative convolution with failure probability $\delta$ (for any $0 \le \delta \le 1/3$). Note that this notation hides the dependence on $\Delta$. 

In the sparse setting, we denote the output size by $k$, i.e., $k$ is the number of nonzero entries of the convolution $A \conv B$. Also in this setting we will always assume that the length $n$ of the input vectors as well as the largest input entry $\Delta$ fit into a constant number of words. We will denote by $S_\delta(k)$ the running time of a randomized algorithm for sparse nonnegative convolution with failure probability~$\delta$; this hides the dependence on $n$ and $\Delta$.

\medskip
The main result of this paper is a novel Monte Carlo algorithm for nonnegative sparse convolution. 

\begin{theorem} \label{thm:corollary}
The sparse nonnegative convolution problem has a randomized algorithm with running time $O(k \log k + \polylog (n\Delta) )$ and failure probability $2^{-\sqrt{ \log k}}$.
\end{theorem}

Naturally, the same algorithm also can be used for Boolean convolution, where $\Delta = 1$. For Boolean convolution, this is the first algorithm that improves upon the dense case's running time of $O(n \log n)$ for all $k = o(n)$; previous algorithms required $k = o(n/\log n)$. For nonnegative convolution, the same statement is true assuming that \raisebox{0pt}[0pt][0pt]{$\Delta \le 2^{n^{o(1)}}$}. Moreover, this answers our driving question for $k \gg \polylog(n \Delta)$, by a randomized algorithm.

In fact, our algorithm can be phrased as a reduction from the sparse case to the dense case of nonnegative convolution:

\begin{theorem} \label{thm:core}
Any randomized algorithm for dense nonnegative convolution with running time $D_{1/3}(n)$ can be turned into a randomized algorithm for sparse nonnegative convolution running in time\footnote{To be precise, we should take the dependence on $\Delta$ (and $n$) into account. Expressing the running time for the dense case as $D_\delta(n,\Delta)$ and for the sparse case as $S_\delta(k,n,\Delta)$, our reduction actually shows that $S_\delta(k,n,\Delta) = O( D_{1/3}(k, \poly(n \Delta)) + k \log^2(\log(k)/\delta) + \polylog(n\Delta))$.}
\begin{equation*}
  S_\delta(k) = O\big( D_{1/3}(k) + k \log^2 (\log(k)/\delta) + \polylog(n\Delta)\big).
\end{equation*}
Here we assume for technical reasons that the function $D_\delta(n) / n$ is nondecreasing, as is to be expected from any natural running time.
\end{theorem}

Since $D_{1/3}(k) = O(k \log k)$, setting $\delta = 2^{-\sqrt{\log k}}$ yields time $O(k \log k + \polylog(n\Delta))$, which proves \autoref{thm:corollary}. Furthermore, any future algorithmic improvement for the dense case automatically yields an improved algorithm for the sparse case by our reduction. In fact, under the mild conditions that $k \gg \polylog(n\Delta)$ and that the optimal running time $D_{1/3}(k)$ is $\Omega(k (\log \log k)^2)$, we obtain an \emph{asymptotic equivalence} with respect to constant-error randomized algorithms: 
\begin{itemize}
\item $S_{1/3}(k) = O(D_{1/3}(k))$ holds by \autoref{thm:core} and the mild conditions,
\item $D_{1/3}(n) = O(S_{1/3}(n))$ holds since the sparse case trivially is a special case of the dense one.
\end{itemize}

\subsection{Discussion and Open Problems}
Our work raises a plethora of open problems that we discuss in the following.

\paragraph*{Improving our Reduction.}
We can ask for improvements of our reduction, specifically of the parameters of \autoref{thm:core}:
\begin{enumerate}[leftmargin=2.4em]
\item Can the error probability of the reduction be reduced? Specifically, can \autoref{thm:corollary} be improved from $\delta = 2^{-\sqrt{\log k}}$ to $1/\poly(k)$ or even $1/\poly(n)$?
\item Can the $\polylog(n \Delta)$ term in \autoref{thm:core} be removed, to make it work also for very small~$k$? This would require a quite different approach than the one we take here, since already for finding a prime field large enough to store $n$ and $\Delta$, or for computing a single multiplicative inverse in such a prime field, the fastest algorithms that we are aware of run in time $O(\polylog(n\Delta))$, even for the Word RAM model. 
\item Can we obtain further improvements by bit packing, say for Boolean convolution? 
\end{enumerate}

\paragraph*{General Convolution.}
Here we focused on nonnegative convolution, what about the general case?
\begin{enumerate}[resume, leftmargin=2.4em]
\item Does sparse general (not necessarily nonnegative) convolution have a randomized algorithm running in time $O(k \log k + \polylog (n\Delta))$?
\item Are sparse general convolution and dense general convolution asymptotically equivalent?
\end{enumerate}

\paragraph*{Deterministic Algorithms.}
Chan and Lewenstein~\cite{ChanL15} presented a deterministic $k \cdot n^{o(1)}$-time algorithm for sparse nonnegative convolution, assuming that they are additionally given a small superset of the output.
\begin{enumerate}[resume, leftmargin=2.4em]
\item Is there a deterministic algorithm for sparse nonnegative convolution with running time $k \polylog(n)$?
\item Are sparse and dense nonnegative convolution asymptotically equivalent with respect to deterministic algorithms?
\end{enumerate}

\paragraph*{Sparse Fourier Transform.}
Computing convolutions is intimately connected to the Fast Fourier Transform (FFT). In fact, in the dense case these two problems are known to be equivalent: if one of these problems can be solved in time $T(n)$ then the other can be solved in time $O(T(n))$. One direction of this equivalence follows from the standard algorithm for convolution that uses FFT, the other direction follows from an old trick invented by Bluestein~\cite{Bluestein70}, see also~\cite[pp.~213--215]{GoldR69}, showing how to express the discrete Fourier transform as a convolution.\footnote{As a technical detail, this reduction assumes that terms of the form $\exp(2\pi i x)$ can be evaluated in constant time.}

The sparse case of Fourier transform, where one has oracle access to $x$ and wants to compute $\widehat{x}$ under a $k$-sparsity assumption, is also extensively studied~\cite{GilbertGIMS02,GilbertMS05,HassaniehIKP12,HavivR17,IndykK14,IndykKP14,PriceS15,Kapralov16,Kapralov17,KapralovVZ19,NakosSW19}. We can ask whether the results presented in this paper also work for computing Fourier transforms:
\begin{enumerate}[resume, leftmargin=2.4em]
\item Can we reduce sparse Fourier transform to dense Fourier transform in a fine-grained way? The algorithm in~\cite{HassaniehIKP12} runs time $O(k \log(n \Delta))$, but the running time is not dominated by the calls to FFT. 
\end{enumerate}
Note that positive answers to Questions 5 and 8 would, together with the known equivalence of dense convolution and dense Fourier transform, show an asymptotic equivalence of sparse general convolution and sparse Fourier transform.

Note that since dense nonnegative convolution is equivalent to general dense convolution (as mentioned in the introduction), and since the latter is equivalent also to (dense) DFT computation, our work places sparse nonnegative convolution to the aforementioned equivalent class, under the assumptions made.

\medskip
We hope that our work ignites further work on revealing connections between all these fundamental problems. 

\subsection{Organization}
This paper is organized as follows. \autoref{sec:preliminaries} starts with some preliminary definitions. In \autoref{sec:tech} we sketch our algorithm and describe some technical difficulties and highlights. The reduction is split across several sections starting with \autoref{sec:tools} which gathers some algorithmic tools, followed by the individual steps of the reduction in Sections~\ref{sec:set-query}--\ref{sec:large-to-small}; we give an outline for these sections in \autoref{sec:tech}. Finally, in \autoref{sec:hashing}, we show an improved concentration bound for linear hashing, which we used as an essential ingredient in our reduction, as well as an almost tight lower bound against a theorem from~\cite{Knudsen16}.


\section{Preliminaries} \label{sec:preliminaries}
Let $\Int$, $\Nat$, $\Rat$ and $\Complex$ to denote the integers, nonnegative integers, rationals, and complex numbers, respectively. For any nonnegative integer $n$, let $\Int_n$ denote the ring of integers modulo $n$. We set $[n] = \{0, 1, \dots, n-1\}$. The Iverson bracket $[P] \in \{0, 1\}$ denotes the truth value of a proposition~$P$. We write $\log$ for the base-$2$ logarithm, $\poly(n) = n^{\Order(1)}$ and $\polylog(n) = \log^{\Order(1)} n$.

We mostly denote vectors by letters~$A, B, C$ with $A_i$ referring to the $i$-th coordinate in $A$. The \emph{convolution} of two length-$n$ vectors $A$ and $B$ is the vector $A \conv B$ of length~$2n-1$ with
\begin{equation*}
	(A \conv B)_i = \sum_{0 \leq j \leq i} A_j B_{i-j}.
\end{equation*}
The \emph{cyclic convolution} of two length-$n$ vectors $A, B$ is the length-$n$ vector $A \conv_n B$ with
\begin{equation*}
	(A \conv_n B)_i = \sum_{0 \leq j \leq n-1} A_j B_{(i-j) \bmod n}.
\end{equation*}
We let $\supp(A) = \{ i \in [n] : A_i \neq 0 \}$,  $\norm A_0 = |\supp(A)|$ and $\norm A_\infty = \max_i |A_i|$. Furthermore, we often write $A \bmod m$ for the vector with
\begin{equation*}
	(A \bmod m)_{i'} = \sum_{i = i' \mod m} A_i,
\end{equation*}
and more generally, for a function $\Int \to [m]$, we write $f(A)$ for the length-$m$ vector with
\begin{equation*}
	f(A)_{i'} = \sum_{i : f(i) = i'} A_i.
\end{equation*}
For sets $X, Y$, we define the \emph{sumset} $X + Y = \{ x + y : (x, y) \in X \times Y\}$ and some other shorthand notation: $a + X = \{ a + x : x \in X \}$, $aX = \{ ax : x \in X \}$, $X \bdiv a = \{ \floor{\frac xa} : x \in X \}$ and $X \bmod a = \{ x \bmod a : x \in X \}$. More generally, for a function defined on $X$ we set $f(X) = \{ f(x) : x \in X \}$.

\section{Technical Overview} \label{sec:tech}
\subsection{Previous Techniques} \label{sec:prevtechniques}
Possibly the earliest work on sparse convolution is a quite complicated $\Order(k \log^2 n + \polylog(n))$-time\footnote{The claimed running time in their paper is $O(k\log^2n)$, however they need to pick a prime $p\in[n,2n]$, which requires time $\polylog(n)$ (this additive overhead disappears if the algorithm is allowed to hardcode $p$).} algorithm due to Cole in Hariharan~\cite{ColeH02} for the nonnegative case. Their approach builds on linear hashing and string algorithms in order to identify $\supp(A\conv B)$, and involves many ideas such as encoding characters with complex entries before applying convolution. The more recent approaches~\cite{Roche08,Roche11,VanDerHoevenL12,VanDerHoevenL13,ArnoldR15,Nakos20,GiorgiGC20} (the last two of which can also solve the general convolution problem) heavily build on hashing modulo a random prime number. This approach suffers from the loss of one log factor due to the density of the primes given by the Prime Number Theorem. Therefore, these approaches seem hopeless of getting time $\Order(k \log k)$, or even time $o(k \log^2 k)$.

On the other hand, a quite different $\Order(k \log^2 (n\Delta))$ algorithm, not explicitly written down as far as we know, is attainable using techniques from the sparse Fourier transform (assuming that complex exponentials can be evaluated in constant time). It has been established in the celebrated work of Hassanieh, Indyk, Katabi and Price~\cite{HassaniehIKP12} that one can recover a $k$-sparse vector $x \in \Complex^n$ in time $\Order(k \log (n\Delta) )$ by only accessing a small subset of its Fourier transform $\widehat x$. This alone might not seem sufficient, but spelling out the details of~\cite{HassaniehIKP12} reveals that the pattern of accesses to $\widehat x$ is a random arithmetic progression of length $\Order(k \log (n\Delta))$. In light of this, one can additionally leverage known techniques from semi-equispaced Fourier transforms~\cite{DuttR93},~\cite[Section~12]{IndykKP14} to obtain a $\Order(k \log^2 (n\Delta))$-time algorithm. The semi-equispaced Fourier transform is a well-studied subfield of computational Fourier transforms, and results from that area show that $s$ equally spaced Fourier coefficients of a length-$n$ and $s$-sparse vector can be computed in time $\Order(s \log (n\Delta))$~\cite[Section~12]{IndykKP14}. Combining this with the algorithm of~\cite{HassaniehIKP12} yields an $\Order(k \log^2 (n\Delta))$-time algorithm for sparse convolution. The inherent reason for this logarithmic blow-up is that going back and forth in Fourier and time domain is more costly in the sparse case than in the dense case. Furthermore, the above algorithm cannot yield a reduction between the sparse and dense convolution (more generally, the approach of~\cite{HassaniehIKP12} cannot yield such an equivalence, as their running time is not dominated by calls to FFT). It is a very interesting open question in that area to show any equivalence between some variant of sparse and dense Fourier transform, as well as to achieve $\Order(k \log k)$ running time.

There are other techniques for sparse convolution using polynomial interpolation, see~\cite{Roche18}, but they do not seem sufficient in going beyond a $\Order(k \polylog n)$-time algorithm in any variation of the problem, owing to the usage of a variety of tools from structured linear algebra which come with additional $\mathrm{polylog}$ factors. 

\subsection{Our Approach}
The goal is to solve the following problem in output-sensitive time:

\begin{restatable}[{{{\normalfont\scshape SparseConv}}}]{problem*}{probsparseconv}
\InputTask
    {Nonnegative vectors $A, B$ and a parameter $\delta > 0$.}
    {Compute $A \conv B$ with success probability $1 - \delta$.}
\end{restatable}

In what follows, we assume that we are given a number $k$ such that $\|A \conv B\|_0 \leq k$, and we want to recover $A \conv B$ in time $\Order(k \log k)$. This assumption will be removed in \autoref{sec:guess-k} using standard techniques. For the sake of simplicity, we will focus on how to obtain a constant-error randomized algorithm for sparse convolution from a deterministic algorithm for dense convolution.

\paragraph{The Obstructions Created by Known Recovery Techniques.}
So far, hashing-based approaches on computing sparse convolutions build on either of two well-known hash functions mapping $[n] \to [m]$:
\begin{itemize}
\item $g(x) = x~\mathrm{mod}~p$, where $p$ is a random prime of appropriate size. 
\item \emph{Linear hashing}: $h(x) = ((\sigma x + \tau) \bmod p) \bmod m$, where $p$ is a sufficiently large fixed prime number and $\sigma, \tau$ are random.\footnote{One can also use $h(x) = \lfloor ((\sigma x + \tau) \bmod p) m / p \big\rfloor$ for a sufficiently large prime $p$, which enjoys similar properties~\cite{Knudsen16}.} 
\end{itemize}
The first hash function satisfies $g(x+y) = (g(x) + g(y)) \bmod m$, in particular it is \emph{affine}, in the sense that $g(x+y) + g(0) \equiv g(x)+g(y) \pmod m$. In comparison, the second hash function is only \emph{almost-affine}, in the sense that $h(x+y) + h(0) - h(x) - h(y)$ can only take a constant number of different values. Although almost-affinity is an amenable issue in many situations, e.g.~\cite{Patrascu10,ChanH20}, in our case it appears to be a more serious obstruction for reasons outlined later.

In turn, the first hash function is only $O(\log n)$-universal. Thus, if we want to hash a size-$k$ set $X$ using $g$, such that a fixed $x \in X$ is isolated from every other $x' \in X$, we must pick $p = \Omega(k \log n)$. This results in a multiplicative $O(\log n)$ overhead on top of the number of buckets. In comparison, linear hashing is $O(1)$-universal, so setting $m=O(k)$ suffices for proper isolation.

\medskip
Before delving deeper, let us sketch how to design an $O(k\log k)$-time algorithm, assuming that we had an ``ideal'' hash function $\iota : [n] \to [m]$ that is $O(1)$-universal and affine, i.e., combines the best of $g(x)$ and $h(x)$. 
Then the hashed convolution could be easily computed as $\iota(A \conv B) = \iota(A) \conv_m \iota(B)$. 
The next ingredient is the derivative operator from~\cite{Huang19}. Defining the vector $\partial A$ with $(\partial A)_i = i \cdot A_i$, and similarly $\partial B$ with $(\partial B)_i = i \cdot B_i$, we have that $ \partial (A\conv B) = (\partial A)\conv B +  A \conv (\partial B)$, which when combined with the ideal hash function $\iota$ gives $\iota( \partial(A\conv B)) = \iota(\partial A) \conv_m \iota(B) +  \iota(A) \conv_m \iota(\partial B)$. The $b$-th coordinate of this vector is
\[	\iota( \partial(A\conv B))_b \;=\; \sum_{i: \iota(i) = b} i \cdot (A\conv B)_i, \]
which can be accessed by computing the length-$m$ convolutions $\iota(\partial A) \conv_m \iota(B)$ and $\iota(A) \conv_m \iota(\partial B)$ and adding them together.
By setting $m = O(k)$, we can now infer a constant fraction of elements $i \in \supp(A\conv B)$ by performing the division 
\[ \frac{\iota(\partial(A\conv B))_b}{\iota((A\conv B))_b} = \frac{\sum_{i: \iota(i) = b} i \cdot (A\conv B)_i}{\sum_{i: \iota(i) = b}  (A\conv B)_i} \] for all $b \in [m]$. This yields the locations of all isolated elements in $\supp(A\conv B)$ under $\iota$. In particular, we obtain a vector $\widetilde C$ such that $\|A\conv B-\widetilde C\|_0 \leq k / 2$, say. 

Now, a classical linear sketching technique~\cite{GilbertLPS10} kicks in. The idea is that we can recover the residual vector $A\conv B- \widetilde C$ by iteratively hashing to a geometrically decreasing number of buckets and performing the same recovery step as before. The number of buckets in the $\ell$-th iteration is $m_\ell = \Order(k / 2^\ell)$, and the goal is to obtain a sequence of vectors $\widetilde{C}^\ell$ such that $\|A\conv B- \widetilde{C}^\ell\|_0 \leq k / 2^\ell$. The crucial observation is that since $\iota$ is affine, we can \emph{cancel out} the contribution of the found elements $\widetilde{C}^\ell$ by the fact that $\iota(A\conv B) -\iota( \widetilde{C}^\ell) = \iota(A\conv B - \widetilde{C}^\ell)$. Thus, after $R=O(\log k)$ iterations~\cite{GilbertLPS10} we obtain a vector $\widetilde{C}^R$ such that $\norm{A\conv B- \widetilde{C}^R}_0 = 0$, recovering $A\conv B$. The running time is dominated by the first iteration, where an FFT over vectors of length $O(k)$ is performed.

Unfortunately, we do not have access to such an ideal function $\iota$. Replacing $\iota$ by $h$ or $g$ runs into issues: If we use $h$ as a substitute, we \emph{cannot cancel out} the contribution of the found elements, since $h$ is only almost affine but not affine. Specifically, the sparsity of $h(A\conv B) - h(\widetilde{C}^\ell)$ does not necessarily decrease in the next iteration, which renders the geometric decreasing number of buckets impossible and thus precludes iterative recovery. If we use $g$ as a substitute, we need to pay additional log factors to ensure isolation of most coordinates, even in the very first iteration.

\medskip
Given this discussion, it seems that the known hash functions reach a barrier on the way to designing $O(k \log k)$-time algorithms. We show how to remedy this state of affairs. 

In the following we describe our approach in five steps.

\paragraph{Step 0: Universe Reduction from Large to Small.}
The first step is to reduce our problem to a universe of size $U = \poly(k)$. We will refer to this regime of $U$ as a \emph{small} universe, and say that $U$ is \emph{large} if there is no bound on $U$. Formally, we introduce the following problem.

\begin{restatable}[{{{\normalfont\SmallUnivSparseConv}}}]{problem*}{probsmallunivsparseconv}
\InputTask
    {Nonnegative vectors $A, B$ of length $U$, an integer $k$ such that $\norm{A \conv B}_0 \leq k$ and $U = \poly(k)$.}
    {Compute $A \conv B$ with success probability $1 - \delta$.}
\end{restatable}

We show in \autoref{sec:large-to-small} how to reduce the general problem of computing $A \conv B$ in a large universe~$n$ to three instances in a small universe $U$. This makes use of the fact that in this parameter regime the linear hash function $h$ is \emph{perfect} with probability $1-1/\poly(k)$. In combination with the derivative operator $\partial$, it suffices to compute the three convolutions $h(A)\conv_U h(B), h(\partial A) \conv_U h(B)$, $h(A) \conv_U h(\partial B)$. Note that the cyclic convolution $\conv_U$ can be reduced in the nonnegative case to the non-cyclic convolution at the cost of doubling the sparsity of the underlying vector, i.e., $\|h(A) \conv h(B)\|_0 \leq 2\|h(A)\conv_U h(B)\|_0 \leq 2\|A \conv B\|_0 \leq 2k$. This yields the claimed reduction.

This universe reduction ensures that from now on the function $g(x) = x \bmod p$ is $O(\log k)$-universal, i.e., we have removed its undesired dependence on $n$, which will be important for the next step. We stress as a subtle detail that this step crucially relies on the fact that we are dealing with nonnegative convolution, for more details see \autoref{sec:generalconvolution}.

\paragraph{Step 1: Error Correction.} 
In the next step, we show that it suffices to compute the convolution $A \conv B$ up to $k / \polylog k$ errors, since we can correct these errors by iterative recovery with the affine hash function~$g$. More precisely assume that we can somehow recover a vector $\widetilde C$ such that $\|A\conv B - \widetilde C\|_0 \leq k / \polylog k$. In other words, suppose that we could efficiently solve the following problem for an appropriate parameter $\gamma$ (think of $\gamma = 1/\log k$).

\begin{restatable}[{{{\normalfont\SmallUnivApproxSparseConv}}}]{problem*}{probsmallunivapproxsparseconv}
\InputTask
    {Nonnegative vectors $A, B$ of length $U$, an integer $k$ such that $\norm{A \conv B}_0 \leq k$ and $U = \poly(k)$.}
    {Compute $\widetilde C$ such that $\norm{A \conv B- \widetilde C}_0 \leq \gamma k$ with success probability $1 - \delta$.}
\end{restatable}

If we are able to do so, then the remaining goal is to correct the error between $A \conv B$ and~$\widetilde C$. We can access the residual vector via $g(A) \conv_m g(B) - g(\widetilde C) = g(A\conv B - \widetilde C)$, for~$g:[U]\rightarrow [O(k)]$. Thus, since the new universe size is a log factor larger than the sparsity of the residual vector, it is possible to continue in an iterative fashion using $g$ and still be within the $O(k \log k)$ time bound. Note that
\begin{enumerate*}[font=,label=(\roman*)]
\item it is crucial that we have recovered a $(1-1/\log k)$-fraction of the coordinates of $\widetilde C$ rather than only a constant fraction, and
\item it can (and will) be the case that $\supp(\widetilde C) \setminus \supp(A\conv B) \neq \emptyset$, i.e., there are spurious elements, but those spurious elements will be removed upon iterating.
\end{enumerate*}

There is one catch: Iterative recovery creates a sequence of successive approximations $\widetilde{C}^1, \widetilde{C}^2,\ldots $ to $A\conv B$, and the time to hash each such vector, i.e., to perform the subtraction $g(A) \conv_m g(B) - g(\widetilde{C}^\ell)$, is $O(k)$. Since there are $O(\log k)$ such subtractions, the total cost spent on subtractions is $O(k \log k)$, which suffices for \autoref{thm:corollary} but not for \autoref{thm:core}. The natural solution is to reduce the number of successive approximations (iterations), which is closely related to the column sparsity of linear sketches that allow iterative recovery. More sophisticated iterative loop invariants exist~\cite{IndykKPW11,PriceW12,GilbertNPRS13}, but these all get $\Omega(\log k)$ column sparsity. What we observe is that, surprisingly, a small modification of the iterative loop in~\cite{GilbertLPS10} finishes in $O(\log \log k)$ iterations, rather than $O(\log k)$. In the $\ell$-th iteration we hash to $O(k / \ell^2)$ buckets, and let $k_\ell = \|A\conv B - \widetilde{C}^\ell\|_0$. An easy argument yields that with probability $1-1/\ell^2$ we have $k_{\ell+1} \leq 1/10 \mult k_\ell^2 / k \mult \ell^4$, which yields $k_L  < 1$ for $L = O( \log \log k)$. This means that the subtraction is performed $O( \log \log k)$ times, so the additive running time overhead is only $O(k \log \log k)$. A more involved implementation of this idea (due to the fact that we are interested in $o(1)$ failure probability) appears in \autoref{sec:error_correction}.

\paragraph{An Attempt using Prony's Method.} 
So far we have reduced to small universe and established that we can afford $k/\polylog k$ errors. In the following we want to recover a $(1-1/\log k)$-fraction of the coordinates ``in one shot''. Consider the following line of attack. Fix a parameter $T \ll k$ and a linear hash function $h : [U] \rightarrow [k/T]$. We aim to recover, for each bucket $b \in [k/T]$, all entries of the convolution $A \conv B$ that are hashed to bucket $b$.\footnote{Here and in the following for ease of exposition we ignore the issue that entries of $A \conv B$ can be split up, due to~$h$ being only almost-affine.} This corresponds to hashing $A \conv B$ to $k/T$ buckets; we expect to have $T$ elements per bucket and thus most buckets contain at most $2T$ elements, say. Note that we no longer expect isolated buckets, so we cannot use the derivative operator. However, we can instead get access to the first~$4T$ Fourier coefficients of each vector $(A\conv B)_{ h^{-1}(b) }$ in the following way. Let $\omega$ be a $U$-th root of unity. For each $t \in [2T]$, set $(\omega^t \bullet A)_i = \omega^{ti} A_i$ and $(\omega^t \bullet B)_i = \omega^{ti} B_i$ and perform the convolution $ h(\omega^t \bullet A) \conv_{k/T} h(\omega^t \bullet B)$. This yields
\begin{align*}
(h(\omega^t \bullet A) \conv_{k/T} h(\omega^t \bullet B))_b = \sum_{(h(i)+h(j)) \bmod k/T = b } \omega^{t(i+j)}\cdot A_i B_j,
\end{align*}
which is essentially the $t$-th Fourier coefficient of $(A\conv B)_{ h^{-1}(b) }$.

The time to perform these $4T$ convolutions is $O((k/T) \cdot \log (k/T)) \cdot 4T = O(k \log k)$. Now, a classical algorithm due to Gaspard de Prony in 1796 (rediscovered several times since then, for decoding BCH codes~\cite{Wolf67} and in the context of polynomial interpolation~\cite{BenOrT88}) postulates that any $2T$-sparse vector can be efficiently reconstructed from its first $4T$ Fourier coefficients. However, Prony's method with finite precision or over a finite field does not have sufficiently fast algorithms for our needs.

Nevertheless, there is another problem with this approach. Since we want to recover a $(1-1/\log k)$-fraction of elements in $A \conv B$, for a $(1-1/\log k)$-fraction of support elements $i \in \supp(A \conv B)$ it must be the case that $|h^{-1}(h(i))| \leq 2T$. This is a necessary condition in order to recover $(A \conv B)_{h^{-1}(h(i))}$ using $4T$ Fourier coefficients. 
If $h$ was three-wise independent, a standard argument using Chebyshev's inequality would show the desired concentration bound. However, since the linear hash function $h$ is only pairwise independent, we need to take a closer look at concentration of linear hashing.

\paragraph{Intermezzo on Linear Hashing.}
A beautiful paper of Knudsen~\cite{Knudsen16} shows that the linear hash function $h$, despite being only pairwise independent, satisfies refined concentration bounds.

\begin{theorem}[Informal Version of {{\cite[Theorem 5]{Knudsen16}}}]
Let $X \subseteq [U]$ be a set of~$k$ keys. Randomly pick a linear hash function $h$ that hashes to $m$ buckets, fix a key $x \not\in X$ and buckets $a, b \in [m]$. Moreover, let $y, z \in X$ be chosen independently and uniformly at random. Then:
\begin{equation}\label{eq:knudsen-informal}
    \Pr(h(y) = h(z) = b \mid h(x) = a) \leq \frac1{m^2} + \frac{2^{\Order(\sqrt{\log k \log\log k})}}{m k}.
\end{equation}
\end{theorem}

Using the above theorem and Chebyshev's inequality, Knudsen arrives at a concentration bound on the number of elements falling in a fixed bucket, see~\cite[Theorem~2]{Knudsen16}.\footnote{We are referring to the FOCS proceedings version, which differs in an important way from the arXiv version.} Up to the factor $2^{\Order(\sqrt{ \log k \log \log k})} = k^{\order(1)}$, this would indeed be the concentration bound satisfied by three-wise independent hash functions. However, this additional $k^{o(1)}$ factor is crucial for our application. Moreover, as we show in \autoref{sec:hashing}, the analysis in~\cite{Knudsen16} is nearly tight. In particular, we show the existence of a set $X$ such that the $k^{o(1)}$ factor is necessary.

\begin{theorem}[{\cite[Theorem 5]{Knudsen16}} is Almost Optimal] \label{thm:linearhashingexample}
Let $k$ and $U$ be arbitrary parameters with $U \geq k^{1+\epsilon}$ for some constant $\epsilon > 0$, and let $h$ be a random linear hash function which hashes to $m$ buckets. Then there exists a set $X \subseteq [U]$ of $k$ keys, a fixed key $x \not\in X$ and buckets $a, b \in [m]$ such that for uniformly random $y, z \in X$ we have
\begin{equation*}
    \Pr(h(y) = h(z) = b \mid h(x) = a) \geq \frac1{mk} \exp\!\left(\Omega\!\left(\sqrt{\min\left(\tfrac{\log k}{\log\log k}, \tfrac{\log U}{\log^2 \log U}\right)}\right)\!\right)\!.
\end{equation*}
\end{theorem}

This brings us to an unclear situation. The structured linear algebra machinery of Prony's method seems inadequate for our purposes and the state of the art concentration bounds of linear hashing do not seem to be sufficiently strong. However, we show again how to remedy this state of affairs. 

\medskip
Our first trick (Step 2) is to reduce to a \emph{tiny} universe of size $k \polylog k$. Note that then \autoref{thm:linearhashingexample} is no longer applicable, and indeed we show improved concentration bounds for linear hashing as we shall see later. Another technical step is to approximate the support of $A \conv B$ (Step 3), which can be done efficiently when the universe is tiny. This replaces the computationally expensive part of Prony's method. After that, we are ready to make the attempt work (Step 4). These steps are described in the following.

\paragraph{Step 2: Universe Reduction from Small to Tiny.}
We further reduce the universe size to $U = k \polylog k$; let us call this regime of $U$ \emph{tiny}. This is the smallest universe we can hash to while ensuring that with constant probability a $(1-1/\log k)$-fraction of coordinates is isolated under the hashing. Apart from this difference the reduction is very similar to Step~0. It remains to solve the following computational problem (again, you may think of $\gamma = 1/ \log k$). This is done in \autoref{sec:small-to-tiny}.

\begin{restatable}[{{{\normalfont\TinyUnivApproxSparseConv}}}]{problem*}{probtinyunivapproxsparseconv}
\InputTask
    {Nonnegative vectors $A, B$ of length $U$, an integer $k$ such that $\norm{A \conv B}_0 \leq k$ and $U \leq k / \gamma^2$.}
    {Compute $\widetilde C$ such that $\norm{A \conv B- \widetilde C}_0 \leq \gamma k$ with success probability $1 - \delta$.}
\end{restatable}

\paragraph{Step 3: Approximating the Support.} 
Next we want to approximate the support $\supp(A \conv B)$. Specifically, we want to recover a set $X$ of size $|X| = \Order(k)$ such that $|\supp(A\conv B) \setminus X| \leq k / \polylog k$. Since $\supp(A \conv B) = \supp(A) + \supp(B)$, for $Y = \supp(A),\, Z = \supp(B)$ we formally want to solve the following problem.

\begin{restatable}[{{{\normalfont\TinyUnivApproxSupp}}}]{problem*}{probtinyunivapproxsupp}
\InputTask
    {Sets $Y, Z \subseteq [U]$ and an integer $k$, such that $U \leq k / \gamma$ and $|Y + Z| \leq k$.}
    {Compute a set $X$ of size $\Order(k)$ such that $|(Y + Z) \setminus X| \leq \gamma k$.}
\end{restatable}

To this end, we create a sequence of successive approximations to $Y + Z$. Consider the sets
\begin{equation*}
    Y_\ell = \left\{ \floor*{\frac y{2^\ell}} : y \in Y \right\}, \quad Z_\ell = \left\{\floor*{\frac z{2^\ell}} : z \in Z \right\},
\end{equation*}
for $0 \le \ell \le \log(U/k)$. For $\ell \ge \log (U/k)$, we have $Y_\ell,Z_\ell \subseteq [k]$, and thus we can compute $X_\ell := Y_\ell + Z_\ell$ by one Boolean convolution in time $O(k \log k)$. Since $U$ is tiny, the number of levels is just $\log(U/k) = O(\log \log k)$. It remains to argue how to go from level~$\ell + 1$ to~$\ell$, to finally approximate $Y_0 + Z_0 = Y + Z$. We say that a set~$X_\ell$ \emph{closely approximates} $Y_\ell + Z_\ell$ if $|X_\ell| = \Order(k)$, and $|(Y_\ell + Z_\ell) \setminus X_\ell| \leq k / \polylog k$. Given a set $X_{\ell+1}$ which closely approximates $Y_{\ell+1} + Z_{\ell+1}$, we want to find a set $X_\ell$ which closely approximates $Y_\ell+Z_\ell$. It is not hard to see that a candidate for~$X_\ell$ is $2 X_{\ell+1} + \{0,1,2\}$. Hence the main problem is keeping the size of~$X_\ell$ small by filtering out false positives. One way to do so would be to compute $h(Y_\ell) + h(Z_\ell)$, for a random linear hash function $h : [U] \to [\Order(k)]$. We then throw away all coordinates $i \in 2X_{\ell+1} + \{0,1,2\}$ for which the bucket $h(i)$ is empty. Naively computing the convolution would lead to time $\Omega(k \log k \log \log k)$. To improve this, we apply an algorithm due to Indyk:

\begin{theorem}[Randomized Boolean Convolution~{{\cite{Indyk98}}}]
There exists an algorithm which takes as input two sets $Y', Z' \subseteq [U]$, and in time $\Order(U)$ outputs a set $\mathcal O \subseteq Y' + Z'$, such that for all $x \in Y' + Z'$ we have $\Pr(x \in \mathcal O) \geq \frac{99}{100}$.
\end{theorem}

Since Indyk's algorithm has a small probability of not reporting an element in the sumset, this leads to losing some of the elements in $\supp(A) + \supp(B)$, but we are fine with $k/\polylog k$ errors. On the positive side, compared to standard Boolean convolution this reduces the running time by a factor $\log k$. 
Putting everything together carefully, we show that $\supp(A\conv B)$ can be approximated in time $\Order(k (\log \log k)^2)$. For the complete proof we refer to Section~\ref{sec:approx-supp}.

\paragraph{Step 4: Approximate Set Query.} 
With all reductions and preparations discussed so far, it remains to solve the following problem to finish our algorithm, for details see \autoref{sec:set-query}.

\noindent
\parbox{\linewidth}{
\begin{restatable}[{{{\normalfont\TinyUnivApproxSetQuery}}}]{problem*}{probtinyunivapproxsetquery}
\InputTask
    {Nonnegative vectors $A, B$ of length $U$, an integer~$k$ such that $\norm{A \conv B}_0 \leq k$ and $U  \leq k / \gamma^2$, and a set $X$ with $|X| = \Order(k)$ and $|\supp(A \conv B) \setminus X| \leq \order(\gamma^2 k)$.}
    {Compute $\widetilde C$ such that $\norm{A \conv B - \widetilde C}_0 \leq \gamma k$ with success probability $1 - \delta$.}
\end{restatable}
}

This is the last step of the algorithm. As in the approach using Prony's method that we discussed above, we pick a parameter~$T$, hash to $k/T$ buckets, and get access to $h(\omega^t \bullet A) \conv_{k/T} h(\omega^t \bullet B)$. Here, $\omega$ is an appropriate element in $\Int_q^\times$ for $q$ a sufficiently large prime. The surprising observation is that in a tiny universe $U$ the lower bound on the concentration bound of linear hashing does not apply, and in fact a much stronger concentration bound is attainable. In particular, we obtain the analogue of \eqref{eq:knudsen-informal} where the term $2^{O(\sqrt{\log k \log \log k})}$ is replaced by $\polylog k$, see \autoref{sec:hashing}. This can be proved using the machinery established in~\cite{Knudsen16} as well as some elementary number theory, and is actually simpler than the complete analysis of~\cite{Knudsen16}.

Furthermore, we can now circumvent the computationally expensive part of Prony's method, since we have knowledge of most of the support $\supp(A \conv B)$. It turns out that we only need to solve $\Order(k/T)$ transposed Vandermonde systems of size $\Order(T)\times\Order(T)$ over~$\Int_q$. The part of the support we do not know might mess up some the estimates due to collisions, but it is such a small fraction that cannot make us misestimate more than a $1/\polylog k$-fraction of the coordinates in $X$ (and the errors that will be introduced due to misestimation will be cleaned up by the iterative recovery loop in Step 2). Using the improved concentration bound for linear hashing, a fast transposed Vandermonde solver~\cite{Li00}, and some additional tricks to compute all $h(\omega^t \bullet A)$ simultaneously, we can pick $T = \polylog k$ and arrive at a $\Order(k \log k)$-time algorithm, that is also a reduction from sparse to dense convolution.

One last detail is that Vandermonde system solvers compute multiplicative inverses, which cost time $\Omega(\log q) = \Omega( \log (n\Delta))$ each, and thus account for time $\Omega(k \log (n\Delta))$ in total. We observe that, since we are solving several (in particular, $k/T$) Vandermonde systems, we can run all of them in parallel and batch the inversions across calls. We can then simulate $k/T$ inversions using $O(k/T)$ multiplications and just one division, see \autoref{lem:bulk-division}. This yields $\Order(k \log k)$ running time and, as claimed in \autoref{thm:core}, an additive $\polylog(n\Delta)$ term (which is already present, only for choosing the prime~$q$).

\subsection{What Makes General Convolution Harder?} \label{sec:generalconvolution}
The reader may ask whether general convolution can be attacked using our techniques. We want to stress that a linear hash function, which is one of our building blocks and at the core of almost all the steps of the algorithm, seems not to be suited for general convolution, due to the fact that it is almost-affine, but not affine. For an element $x \in [n]$ consider the quantities
\begin{align*}
    c_1 &= \sum_{\substack{y + z = x, \\ h(y) + h(z) \equiv h(0) + h(x)}} A_y B_z, \\
    c_2 &= \sum_{\substack{y + z = x, \\ h(y) + h(z) \equiv h(0) + h(x) + p}} A_y B_z, \\
    c_3 &= \sum_{\substack{y + z = x, \\ h(y) + h(z) \equiv h(0) + h(x) - p}} A_y B_z,
\end{align*}
where, for convenience we write $\equiv$ for equality modulo $m$, and $p, m$ are parameters of the linear hash function. By the almost-affinity of linear hashing we have $(A\conv B)_x = c_1+c_2+c_3$ (see \autoref{lem:linear-hashing-basics}). In general, it can happen that $(A\conv B)_x =0$, not contributing at all to the output size, whereas $c_1, c_2, c_3 \neq 0$. This means that what is hashed to $m$ buckets is a vector with sparsity much larger than~$k$. Handling the presence of cancellations in $A \conv B$ is a significant obstruction to an $\Order(k \log k)$ general convolution algorithm. Note that even Step~0 is non-trivial to implement for general convolution.

Unless one can somehow handle this issue, we can only work with $g(x) = x \bmod p$, which comes with additional log factor losses. We believe that a very different approach is needed to obtain time $\Order(k \log k)$ in the general case, which is a very interesting open question.


\section{Tools} \label{sec:tools}
\subsection{Linear Hashing}
In many of our algorithms, the goal is to reduce the dimension of some vectors in a convolution-preserving way. To that end, we often use the classic textbook hash function
\begin{equation*} \label{eq:linear-hashing}
    h(x) = ((\sigma x + \tau) \bmod p) \bmod m.
\end{equation*}
In our case $p$ is always some (fixed) prime, $m \leq p$ is the (fixed) number of buckets and $\sigma, \tau \in [p]$ are chosen uniformly and independently at random. We say that $h$ is a \emph{linear hash function} with parameters $p$ and $m$. We start with some well-known fundamental properties of linear hashing:

\begin{lemma}[Linear Hashing Basics] \label{lem:linear-hashing-basics}
Let $h$ be a linear hash function with parameters~$p$ and~$m$ drawn uniformly at random. Then the following properties hold:
\begin{description}
\item[Universality:] For distinct keys $x, y$ and $a \in [m]$:\\$|\, \Pr(h(x) = h(y) + a \mod m) - \frac1m \,| \leq \frac3p \leq \frac3m$.
\item[Pairwise Independence:] For distinct keys $x, y$ and arbitrary buckets $a, b \in [m]$:\\$|\, \Pr(h(x) = a \land h(y) = b) - \frac1{m^2} \,| \leq \frac3{mp} \leq \frac3{m^2}$. 
\item[Almost-Affinity:] For arbitrary keys $x, y$ there exists one out of three possible offsets $o \in \{-p, 0, p\}$ such that $h(x) + h(y) = h(0) + h(x + y) + o \mod m$.
\end{description} 
\end{lemma}
\begin{proof}
Universality follows directly from pairwise independence, so we start proving pairwise independence. Let $h(x) = \pi(x) \bmod m$, where $\pi(x) = (\sigma x + \tau) \bmod p$ for uniformly random $\sigma, \tau \in [p]$. The first step is to prove that $\Pr(\pi(x) = a' \land \pi(y) = b') = 1/p^2$ for distinct keys $x, y$ and arbitrary $a', b' \in [p]$. Note that the event $\pi(x) = a'$ and $\pi(y) = b'$ can be rewritten as $\pi(x) = a'$ and $\pi(y) - \pi(x) = b' - a' \mod p$ and the claim follows immediately by observing that the random variables $\pi(x)$ and $\pi(y) - \pi(x) = \sigma(y - x) \mod p$ are independent.

We get back to $h$. Clearly, $\Pr(h(x) = a \land h(y) = b) = \sum_{a', b'} \Pr(\pi(x) = a' \land \pi(y) = b')$, where the sum is over all $a', b' \in [p]$ with $a = a' \bmod m$ and $b = b' \bmod m$. As there are at least $\floor{p/m}$ and at most $\ceil{p/m}$ such values~$a'$ and~$b'$, respectively, we conclude that the desired probability is at least $\floor{p/m}^2 / p^2 \geq (p/m - 1)^2 / p^2 \geq \frac1{m^2} - \frac2{mp}$ and at most $\ceil{p/m}^2 / p^2 \leq (p/m + 1)^2 / p^2 \leq \frac1{m^2} + \frac3{mp}$.

Finally, we prove that $h$ is almost-affine. It is clear that $\pi(x) + \pi(y) = \pi(0) + \pi(x + y) \mod p$. As each side of the equation is a nonnegative integer less than $2p$, it follows that $\pi(x) + \pi(y) = \pi(0) + \pi(x + y) + o$, where $o \in \{-p, 0, p\}$. By taking residues modulo $m$, the claim follows.
\end{proof}

In our reduction we also crucially rely on the following improved concentration bound for linear hashing, which we prove in \autoref{sec:hashing}.

\begin{restatable*}[Overfull Buckets]{corollary}{coroverfullbuckets} \label{cor:overfull-buckets}
Let $X \subseteq [U]$ be a set of $k$ keys. Randomly pick a linear hash function $h$ with parameters $p > 4U^2$ and $m \leq U$, fix a key $x \not\in X$ and buckets $a, b \in [m]$. Moreover, let~$F = \sum_{y \in X} [h(y) = b]$. Then:
\begin{equation*}
    \Ex(F \mid h(x) = a) = \Ex(F) = \frac km \pm \Theta(1),
\end{equation*}
and, for any $\lambda > 0$,
\begin{equation*}
    \Pr(\, |F - \Ex(F)| \geq \lambda \sqrt{\Ex(F)} \mid h(x) = a) \leq \Order\!\left(\frac{U \log U}{\lambda^2 k}\right)\!.
\end{equation*}
\end{restatable*}

\subsection{Algebraic Computations}
On more than one occasion we need to efficiently perform algebraic computations such as computing powers or inverses. The next two lemmas describe how to easily obtain improved algorithms for bulk-evaluation.

\begin{lemma}[Bulk Exponentiation] \label{lem:bulk-exponentiation}
Let $R$ be a ring. Given an element $x \in R$, and a set of nonnegative exponents $e_1, \dots, e_n \leq e$, we can compute $x^{e_1}, \dots, x^{e_n}$ in time $\Order(n \log_n e)$ using $\Order(n \log_n e)$ ring operations.
\end{lemma}

The naive way to implement exponentiations is via repeated squaring in time $\Order(n \log e)$. There are methods~\cite{Yao76,Pippenger80} improving the dependence on $e$, but for our purposes this simple algorithm suffices. 

\begin{proof}
First, compute the base-$n$ representations of all exponents \raisebox{0pt}[0pt][0pt]{$e_i = \sum_j e_{i,j} n^j$}; then \raisebox{0pt}[0pt][0pt]{$e_{i,j} \in [n]$} where $j = 0, \dots, \ceil{\log_n e}$. We will precompute all powers $x^{in^j}$ for $i = 1, \dots, n$ and $j = 0, \dots, \ceil{\log_n e}$ using the rules $x^{n^{j+1}} = (x^{n^j})^n$ and $x^{(i+1)n^j} = x^{in^j} x^{n^j}$. Finally, every output $x^{e_i}$ can be computed as a product of $\ceil{\log_n e}$ numbers \raisebox{0pt}[0pt][0pt]{$\prod_j x^{e_{i,j} n^j}$}. The correctness is immediate and it is easy to check that every step takes time $\Order(n \log_n e)$.
\end{proof}

\begin{lemma}[Bulk Division] \label{lem:bulk-division}
Let $F$ be a field. Given $n$ field elements $a_1, \dots, a_n \in F$, we can compute their inverses $a_1^{-1}, \dots, a_n^{-1} \in F$ in time $\Order(n)$ using $\Order(n)$ multiplications and a single inversion.
\end{lemma}
\begin{proof}
First, we compute the $n$ prefix products $b_j = a_1 \cdots a_i$. It takes a single inversion to compute~$b_n^{-1}$. Then, for $i = n, n-1, \dots, 2$, we compute $a_i^{-1} = b_i^{-1} b_{i-1}$ and $b_{i-1}^{-1} = b_i^{-1} a_i$. Finally, \raisebox{0pt}[0pt][0pt]{$a_1^{-1} = b_1^{-1}$}. As claimed, this algorithm takes time $\Order(n)$ and it uses $\Order(n)$ multiplications and a single inversion.
\end{proof}

Finally, a crucial ingredient to our core algorithm is the following theorem about solving transposed Vandermonde systems.

\begin{restatable}[Transposed Vandermonde Systems]{theorem}{thmvandermonde} \label{thm:vandermonde}
Let $F$ be a field. Let $a, x \in F^n$ be vectors with pairwise distinct entries $a_i$, and let
\begin{equation*}
    V = \begin{bmatrix} 1 & 1 & \cdots & 1 \\ a_1 & a_2 & \cdots & a_n \\ a_1^2 & a_2^2 & \cdots & a_n^2 \\ \vdots & \vdots & \ddots & \vdots \\ a_1^{n-1} & a_2^{n-1} & \cdots & a_n^{n-1} \end{bmatrix}\!.
\end{equation*}
Then $V x$ and $V^{-1} x$ can be computed in time $\Order(n \log^2 n)$ using at most $\Order(n \log^2 n)$ ring operations (additions, subtractions and multiplications) and at most 1 division.
\end{restatable}

This algorithm has been discovered several times~\cite{KaltofenY88,Li00,Pan01}. Although none of these sources pays attention to the number of divisions, one can check that applying the bulk division strategy from \autoref{lem:bulk-division} suffices to obtain the claimed bound. For the sake of completeness, we give a full proof of \autoref{thm:vandermonde} in \autoref{sec:vandermonde}.


\section{Set Queries in a Tiny Universe} \label{sec:set-query}
As the first step in our chain of reductions, the goal of this section is to give an efficient algorithm for the \TinyUnivApproxSetQuery{} problem:

\probtinyunivapproxsetquery*{}

\begin{lemma} \label{lem:tiny-univ-set-query}
Let $\log k \leq 1/\gamma \leq \poly(k)$ and let $1/\delta \leq \poly(k)$. There is an algorithm for the \TinyUnivApproxSetQuery{} problem running in time
\begin{equation*}
    \Order(D(k) + k \log^2(1/\gamma) + k \log(1/\delta) + \polylog(k, \norm A_\infty, \norm B_\infty)).
\end{equation*}
\end{lemma}

We proceed in three steps. In \autoref{sec:set-query:sec:folding} we give two important preliminary lemmas. In \autoref{sec:set-query:sec:algorithm} we present and analyze the algorithm which proves \autoref{lem:tiny-univ-set-query}. Finally, in \autoref{sec:algorithm:sec:randomized}, we will strengthen \autoref{lem:tiny-univ-set-query} and show that we can in fact achieve the same running time with $D_{1/3}(k)$ in place of $D(k)$, i.e., it suffices to assume that we only have black-box access to an efficient dense convolution algorithm with constant error probability.

\subsection{Folding \& Unfolding} \label{sec:set-query:sec:folding}
For a vector $A$ and a scalar $\omega$, let $\omega \bullet A$ denote the vector defined by $(\omega \bullet A)_i = \omega^i A_i$. A straightforward calculation reveals that the $\bullet$-product commutes nicely with taking (non-cyclic) convolutions:

\begin{proposition} \label{prop:bullet-product}
Let $A, B$ be vectors and let $\omega$ be a scalar. Then $(\omega \bullet A) \conv (\omega \bullet B) = \omega \bullet (A \conv B)$.    
\end{proposition}
\begin{proof}
For any coordinate $x$:
\begin{equation*}
    ((\omega \bullet A) \conv (\omega \bullet B))_x = \sum_{y + z = x} (\omega \bullet A)_y (\omega \bullet B)_z = \sum_{y + z = x} \omega^{y + z} A_y B_z = (\omega \bullet (A \conv B))_x. \qedhere
\end{equation*}
\end{proof}

The goal of this section is to show that the we can efficiently evaluate, and under certain restrictions also invert, the following map:
\begin{equation*}
    A \;\longrightarrow\; (\omega^0 \bullet A) \bmod m, \dots, (\omega^{T-1} \bullet A) \bmod m.
\end{equation*}
We will vaguely refer to these two directions as \emph{folding} and \emph{unfolding}, respectively. For the remainder of this subsection we assume as before that~$A$ is an arbitrary length-$U$ vector with sparsity~$k$. We further assume that~$A$ is over some finite field~$\Int_q$ in order avoid precision issues in the underlying algebraic machinery. We also need the technical assumption that $\omega \in \Int_q^\times$ has multiplicative order at least $U$.

\begin{lemma}[Folding] \label{lem:folding}
Let $m$ be a parameter and let $T = \ceil{2k/m}$. There is a deterministic algorithm \Fold{} computing $(\omega^0 \bullet A) \bmod m, \dots, (\omega^{T-1} \bullet A) \bmod m$ in time $\Order(k \log^2(k/m) + k \log_k U + \polylog q)$.
\end{lemma}

Let us postpone the proof of \autoref{lem:folding} and instead outline how to (approximately) invert the folding. A crucial assumption is that we are given a close approximation $X$ of $\supp(A)$. The quality of the recovery is controlled by the following measure: The \emph{flatness of~$X$ modulo~$m$} is defined as
\begin{equation*}
    F_m(X) = \sum_{x \in X} \left[ \sum_{x' \in X} \left[x = x' \mod m\right] > \frac{2|X|}m \right]\!,
\end{equation*}
and we say that $X$ is \emph{$\alpha$-flat modulo $m$} if $F_m(X) \leq \alpha$. Some intuition about this definition: Recall that when hashing a set $X$ into $m$ buckets, the average bucket receives $|X|/m$ elements. Therefore, the flatness is the number of elements falling into overfull buckets under the very simple hash function $x \bmod m$.

\begin{lemma}[Unfolding] \label{lem:unfolding}
Let $m$ be a parameter and let $T = \ceil{2k/m}$. There is a deterministic algorithm \Unfold{} which, given $(\omega^0 \bullet A) \bmod m, \dots, (\omega^{T-1} \bullet A) \bmod m$ and a size-$k$ set $X \subseteq [U]$, computes a vector $\widetilde A$ such that
\begin{equation*}
    \norm{A - \widetilde A}_0 \leq T \mult |\supp(A) \setminus X| + F_m(X).
\end{equation*}
The algorithm runs in time $\Order(k \log^2(k/m) + k \log_k U + \polylog q)$.
\end{lemma}

We will next prove Lemmas~\ref{lem:folding} and~\ref{lem:unfolding}.

\begin{proof}[Proof of \autoref{lem:folding}]
Given a $k$-sparse vector $A$ the goal is to simultaneously compute $A^0 = (A \bullet \omega^0) \bmod m, \dots, A^{T-1} = (A \bullet \omega^{T-1}) \bmod m$. We first precompute all powers $\omega^x$ for $x \in X = \supp(A)$ using the bulk exponentiation algorithm (\autoref{lem:bulk-exponentiation}).

Next, we partition $X$ into several \emph{chunks} $X_{i,j}$. We start with $X_i = \{x \in X : x \bmod m = i\}$ and then greedily subdivide every part $X_i$ into chunks $X_{i,1}, X_{i,2}, \dots$ such that all chunks have size $|X_{i,j}| \leq T$. In fact, all chunks except for the last one have size exactly $T$. We note that in this way we have constructed at most $\Order(m)$ chunks: On the one hand, there can be at most $m$ chunks of size exactly~$T$ since $A$ has sparsity $k = \Theta(m T)$. On the other hand, there can be at most $m$ chunks of size less than~$T$ by the way the greedy algorithm works.

Now focus on an arbitrary chunk $X_{i,j}$; for simplicity assume that $|X_{i,j}| = T$ and let $x_1, \dots, x_T$ denote the elements of $X_{i,j}$ in an arbitrary order. We set up the following transposed Vandermonde system with indeterminate $y^{i,j} \in \Int_q^T$:
\begin{equation*}
    y^{i,j} = \begin{bmatrix} 1 & 1 & \cdots & 1 \\ \omega^{x_1} & \omega^{x_2} & \cdots & \omega^{x_T} \\ \omega^{2x_1} & \omega^{2x_2} & \cdots & \omega^{2x_T} \\ \vdots & \vdots & \ddots & \vdots \\ \omega^{(T-1)x_1} & \omega^{(T-1)x_2} & \cdots & \omega^{(T-1)x_T} \end{bmatrix} \begin{bmatrix} A_{x_1} \\ A_{x_2} \\ \vdots \\ A_{x_T} \end{bmatrix}.
\end{equation*}
Since $\omega$ has multiplicative order at least $U$, the coefficients $\omega^{x_1}, \dots, \omega^{x_T}$ are distinct and we can apply \autoref{thm:vandermonde} to compute $y$. It remains to return the vectors $(\omega^t \bullet A) \bmod m$ for all~$t \in [T]$, computed as \raisebox{0pt}[0pt][0pt]{$((\omega^t \bullet A) \bmod m)_i = \sum_j y^{i,j}_t$}. It is easy to check that \raisebox{0pt}[0pt][0pt]{$y^{i,j}_t$} equals $((\omega^t \bullet A') \bmod m)_i$ for~$A'$ the vector obtained from $A$ by restricting the support to $X_{i,j}$. The correctness of the whole algorithm follows immediately.

Finally, we analyze the running time. Precomputing the powers of $\omega$ using \autoref{lem:bulk-exponentiation} accounts for time $\Order(k \log_k U)$. The construction of the chunks takes time $\Order(m T) = \Order(k)$, and also writing down all vectors $(\omega^t \bullet A) \bmod m$ takes time $\Order(k)$ given the $y^{i,j}$'s. The dominant step is to solve a Vandermonde system for every chunk. Since there are $\Order(m)$ chunks in total and the running time for solving a single system is bounded by $\Order(T \log^2 T)$ (by \autoref{thm:vandermonde}), the total running time is $\Order(m T \log^2 T)$ plus $\Order(m T \log^2 T)$ ring operations and $\Order(m\polylog(T))$ divisions in~$\Int_q$.

Additions, subtractions and multiplications take constant time each on a random-access machine and can therefore by counted into the time bound. However, divisions in a prime field are computationally more expensive. The common way is to implement inversions by Euclid's algorithm in time $\Order(\log q)$ and so the naive time bound becomes $\Order(m \polylog(T) \mult \log q)$. This can be optimized by exploiting \autoref{lem:bulk-division}: Recall that we are executing \autoref{thm:vandermonde} $m$ times in parallel, and each call requires up to $\polylog(T)$ inversions. Therefore, we can apply \autoref{lem:bulk-division} to replace~$m$ inversions by $\Order(m)$ multiplications and a single inversion in time $\Order(m + \log q)$. In that way, it takes time $\Order(\polylog(T) \mult (m + \log q)) = \Order(m \polylog(T) + \polylog(q))$ to deal with all divisions and the total running time is $\Order(m T \log^2 T + \polylog(q)) = \Order(k \log^2 (m/k) + \polylog(q))$.
\end{proof}

\begin{proof}[Proof of \autoref{lem:unfolding}]
Given $A^0 = (A \bullet \omega^0) \bmod m, \dots, A^{T-1} = (A \bullet \omega^{T-1}) \bmod m$, the goal is to recover a good approximation $\widetilde A$ of $A$, provided that an approximation $X$ of $\supp(A)$ is given. As before, we partition $X$ into buckets $X_i = \{ x \in X : x \bmod m = i \}$. We say that the bucket $X_i$ is \emph{overfull} if $|X_i| > T$. In contrast to \Fold{}, we can afford to ignore all overfull buckets here, so focus on an arbitrary bucket $X_i$ with $|X_i| \leq T$. Letting $x_1, \dots, x_T$ denote the elements in $X_i$ in an arbitrary order (and assuming for the sake of simplicity that there are exactly $T$ of these), it suffices to solve the following Vandermonde system with indeterminates $\widetilde A_{x_1}, \dots, \widetilde A_{x_T}$:
\begin{equation*}
    \begin{bmatrix} A^0_i \\ A^1_i \\ A^2_i \\ \vdots \\ A^{T-1}_i \end{bmatrix} = \begin{bmatrix} 1 & 1 & \cdots & 1 \\ \omega^{x_1} & \omega^{x_2} & \cdots & \omega^{x_T} \\ \omega^{2x_1} & \omega^{2x_2} & \cdots & \omega^{2x_T} \\ \vdots & \vdots & \ddots & \vdots \\ \omega^{(T-1)x_1} & \omega^{(T-1)x_2} & \cdots & \omega^{(T-1)x_T} \end{bmatrix} \begin{bmatrix} \widetilde A_{x_1} \\ \widetilde A_{x_2} \\ \vdots \\ \widetilde A_{x_T} \end{bmatrix}.
\end{equation*}
The running time can be analyzed in the same way as before, so let us focus on proving that $\norm{A - \widetilde A}_0$ is small. We say that a bucket $X_i$ is \emph{successful} if
\begin{enumerate*}[font=, label=(\roman*)]
\item\label{lem:folding-unfolding:itm:overfull} it is not overfull, and if
\item\label{lem:folding-unfolding:itm:approx} there exists no support element $x \in \supp(A) \setminus X$ with $x \bmod m = i$.
\end{enumerate*}
The claim is that whenever $X_i$ is successful, then $\widetilde A_x = A_x$ for all $x \in X_i$. Indeed, for any successful bucket one can verify by the definition of the $\bullet$-product that the equation system is valid for $A_x$ in place of $\widetilde A_x$, and as the Vandermonde matrix has full rank this is the unique solution.

Therefore, it suffices to bound the total size of all non-successful buckets: On the one hand, the number of elements in buckets for which condition~\ref{lem:folding-unfolding:itm:overfull} holds but~\ref{lem:folding-unfolding:itm:approx} fails is at most $T \mult |\supp(A) \setminus X|$. On the other hand, the contribution of elements in buckets for which condition~\ref{lem:folding-unfolding:itm:overfull} fails is exactly the flatness of $X$ modulo $m$, by definition. Together, these yield the claimed bound on $\norm{A - \widetilde A}_0$.
\end{proof}

\subsection{The Algorithm} \label{sec:set-query:sec:algorithm}
We are ready to prove \autoref{lem:tiny-univ-set-query} by analyzing the pseudo-code given in \autoref{alg:set-query}. The analysis is split into three parts corresponding to the three parts in \autoref{alg:set-query}. The first step is to prove that the loop in Part 1 quickly terminates.

\begin{algorithm}[t]
\caption{$\TinyUnivApproxSetQuery(A, B, U, k, X)$} \label{alg:set-query}
\begin{algorithmic}[1]
\Input{
\begin{itemize}[topsep=.1ex, noitemsep, leftmargin=1em]
    \item Nonnegative vectors $A, B$ of length $U \leq k / \gamma^2$
    \item An integer $k$ such that $\norm{A \conv B}_0 \leq k$
    \item A set $X \subseteq [U]$ of size $\Order(k)$ such that $|\supp(A \conv B) \setminus X| \leq \order(\gamma^2 k)$
\end{itemize}}
\smallskip
\Output{A vector $\widetilde C$ such that $\norm{A \conv B - \widetilde C}_0 \leq \gamma k$ with probability $1 - \delta$}

\smallskip
\Statex \emph{(Part 1: Find a suitable linear hash function)}
\State Let $m = \Theta(\gamma k)$, let $T = \floor{2|X| / m}$ and let $p \geq 4U^2$ be a prime
\Repeat \label{line:hash-repeat}
    \State Pick $\sigma, \tau \in [p]$ uniformly at random
    \State Let $\pi(x) = (\sigma x + \tau) \bmod p$
    \State $\widetilde X \gets \pi(X) + \{0, p\}$
\Until {\raisebox{0pt}[0pt][0pt]{$\widetilde X$} is $\gamma k/2$-flat modulo $m$} \label{line:hash-until}
\medskip

\Statex \emph{(Part 2: Set up a sufficiently large finite field)}
\State Let $q > U^3 \norm A_\infty \norm B_\infty$ be a prime; the following calculations are over $\Int_q$ \label{line:q}
\State Pick $\omega \in \Int_q^\times$ uniformly at random \label{line:omega}
\medskip

\Statex \emph{(Part 3: Fold -- Convolve -- Unfold)}
\State $A^0, \dots, A^{T-1} \gets \Fold(\pi(A), \omega)$ \label{line:fold-A}
\State $B^0, \dots, B^{T-1} \gets \Fold(\pi(B), \omega)$ \label{line:fold-B}
\For {$t \gets 0, \dots, T-1$}
    \State $C^t \gets A^t \conv_m B^t$ (using the dense convolution algorithm) \label{line:conv}
\EndFor \vspace*{-.7ex}
\State $\widetilde R \gets \Unfold(C^0, \dots, C^{T-1}, \omega, \widetilde X)$ \label{line:unfold}
\State $\widetilde C \gets \pi^{-1}(\widetilde R)$
\State\Return \raisebox{0pt}[0pt][0pt]{$\widetilde C$} (cast back to an integer vector)
\smallskip

\end{algorithmic}
\end{algorithm}

\begin{lemma}[Analysis of Part 1] \label{lem:set-query-part-1}
With probability $1 - \delta/2$, the loop in Lines~\ref{line:hash-repeat}--\ref{line:hash-until} terminates in time $\Order(k \log(1/\delta))$.
\end{lemma}
\begin{proof}
We prove that a single iteration of the loop succeeds with constant probability. Having established that fact, it is clear that the loop is left after at most $\Order(\log(1/\delta))$ independent iterations with probability at least $1 - \delta/2$. Recall that the loop ends if $\widetilde X$ is $\gamma k/2$-flat modulo $m$, that is,
\begin{equation} \label{eq:event-tilde}
    \sum_{x \in \widetilde X} \left[ \sum_{x' \in \widetilde X} \left[x = x' \mod m\right] > \frac{2|\widetilde X|}m \right] \leq \frac{\gamma k}2.
\end{equation}
Since by definition $\widetilde X = \pi(X) + \{0, p\} $, we may fix offsets $o, o' \in \{0, p\}$ and instead bound
\begin{equation} \label{eq:event}
    \sum_{x \in X} \left[ \sum_{x' \in X} \left[h(x) + o = h(x') + o' \mod m\right] > \frac{2|X|}m \right] \leq \frac{\gamma k}4,
\end{equation}
where $h(x) = \pi(x) \bmod m$ is a linear hash function with parameters $p$ and $m$. Indeed, if the latter event happens (simultaneously for all offsets~$o, o'$), then also the former event happens. Fix $o, o'$ and fix any $x \in X$. Then:
\begin{align*}
    &\Pr\left( \sum_{x' \in X} \left[h(x) + o = h(x') + o' \mod m\right] > \frac{2|X|}m \right) \\
    ={} &\sum_{a \in [m]} \Pr(h(x) = a) \mult \Pr\left( \sum_{x' \in X} \left[h(x') = (a + o - o') \bmod m\right] > \frac{2|X|}m \;\middle|\; h(x) = a \right) \\
\intertext{This is where our concentration bounds come into play: Observe that the conditional probability can be bounded by \autoref{cor:overfull-buckets} with buckets $a$ and $b = (a + o - o') \bmod m$. Let $F = \sum_{x' \in X} [h(x') = b]$, then $\Ex(F) = |X|/m + \Order(1)$. It follows that:}
    ={} &\sum_{a \in [m]} \Pr(h(x) = a) \mult \Pr\left( F > \frac{2|X|}m \;\middle|\; h(x) = a \right), \\
    ={} &\sum_{a \in [m]} \Pr(h(x) = a) \mult \Pr\left( F - \Ex(F) > \frac{|X|}m - \Order(1) \;\middle|\; h(x) = a \right), \\
    \leq{} &\sum_{a \in [m]} \Pr(h(x) = a) \mult \Order\left(\frac{m U \log U}{|X|^2}\right) \\
    ={} &\Order\left(\frac{m U \log U}{|X|^2}\right)
\intertext{where for the inequality we applied \autoref{cor:overfull-buckets} with $\lambda = \sqrt{|X| / m} - \Order(1)$. We choose $m = \epsilon \gamma k$ for some small constant $\epsilon > 0$, then:}
    ={} &\Order\left(\frac{\epsilon \gamma k U \log U}{k^2}\right) = \Order\left(\frac{\epsilon \gamma k (k / \gamma^2) \log (k / \gamma^2)}{k^2}\right) = \Order\left(\frac{\epsilon \log k}{\gamma}\right) = \Order(\epsilon).
\end{align*}
Here we used the assumption $\log k \leq 1/\gamma \leq \poly(k)$. By setting $\epsilon$ small enough, this probability becomes less than~$1/12$. Then, by a union bound over the three possible values of $o - o'$ and by Markov's inequality, we conclude that the event in~\eqref{eq:event} (and thereby the event in~\eqref{eq:event-tilde}) happens with probability at least $1/2$.

As we just proved, with probability $1 - \delta / 2$ the loop in Lines~\ref{line:hash-repeat}--\ref{line:hash-until} runs for at most $\Order(\log(1/\delta))$ iterations. Moreover, each execution of the loop body takes time $\Order(k)$, and thus the loop terminates in time $\Order(k \log(1/\delta))$.
\end{proof}

\begin{lemma}[Analysis of Part 2] \label{lem:set-query-part-2}
With probability $1 - 1/\poly(k)$ and in $\polylog(k, \norm A_\infty, \norm B_\infty)$ time we correctly compute~$q$ and~$\omega$ such that~$\omega$ has multiplicative order at least $U$ in~$\Int_q^\times$.
\end{lemma}
\begin{proof}
Computing $q$ takes time $\polylog(k, \norm A_\infty, \norm B_\infty)$ and succeeds with high probability. The interesting part is to show that~$\omega$ is as claimed. It is well-known that $\Int_q^\times$ is isomorphic to the cyclic group $\Int_{q-1}$ and thus our problem is equivalent to finding an element in $\Int_{q-1}$ with (additive) order at least $U$. In a cyclic group there can be at most $i$ elements with order $i$ (the only possible candidates are multiples of $(q - 1) / i$) and thus there are at most $\sum_{i \leq U} i \leq U^2$ elements with order at most $U$. Hence, the probability of sampling $\omega$ as claimed is at least $1 - U^2 / q \geq 1 - 1 / U \geq 1 - 1 / \poly(k)$.
\end{proof}

\begin{lemma}[Analysis of Part 3] \label{lem:set-query-part-3}
With probability $1 - 1/\poly(k)$, Part 3 correctly outputs a vector~$\widetilde C$ with $\norm{A \conv B - \widetilde C}_0 \leq \gamma k$ and runs in time $\Order(D(k) + k \log^2(1/\gamma) + \polylog(k, \norm A_\infty, \norm B_\infty))$.
\end{lemma}
\begin{proof}
    In the event that the previous parts succeeded the technical condition of Lemmas~\ref{lem:folding} and~\ref{lem:unfolding} is satisfied (namely that $\omega$ has large multiplicative order) and we may apply \Fold{} and \Unfold{}. In Lines~\ref{line:fold-A} and~\ref{line:fold-B} we thus correctly compute $A^t = (\omega^t \bullet \pi(A)) \bmod m$, for all $t \in [T]$, and similarly for~$B$. As we are assuming (for now) that the dense convolution algorithm succeeds with probability~$1$, in \autoref{line:conv} we correctly compute the cyclic convolutions $C^t = A^t \conv_m B^t$.

The interesting step is to analyze the unfolding in \autoref{line:unfold}. By \autoref{prop:bullet-product} and some elementary identities about cyclic convolutions we have
\begin{align*}
    C^t &= A^t \conv_m B^t \\
    &= (\omega^t \bullet \pi(A)) \conv_m (\omega^t \bullet \pi(B)) \\
    &= ((\omega^t \bullet \pi(A)) \conv (\omega^t \bullet \pi(B))) \bmod m \\
    &= (\omega^t \bullet (\pi(A) \conv \pi(B))) \bmod m,
\end{align*}
i.e.\ it holds that $C^t = (\omega^t \bullet R) \bmod m$ for $R = \pi(A) \conv \pi(B)$. For that reason, \autoref{lem:unfolding} guarantees that the call to \Unfold{} will approximately recover $R$ and the approximation quality is bounded by
\begin{equation*}
    \norm{R - \widetilde R}_0 \leq T \mult |\supp(R) \setminus \widetilde X| + F_m(\widetilde X).
\end{equation*}
By the loop guard in \autoref{line:hash-until} we can assume that $F_m(\widetilde X) \leq \gamma k / 2$. We can put the same bound on the term $T \mult |\supp(R) \setminus \widetilde X|$. Indeed, note that since $\supp(R) \subseteq \pi(\supp(A \conv B)) + \{0, p\}$ and $\widetilde X = X + \{0, p\}$, we must have that $|\supp(R) \setminus \widetilde X| \leq 2 |\supp(A \conv B) \setminus X|$. It follows that
\begin{equation*}
    T \mult |\supp(R) \setminus \widetilde X| \leq 2 T \mult |\supp(A \conv B) \setminus X| \leq \frac{\order(\gamma^2 k^2)}m,
\end{equation*}
which becomes $\gamma k / 2$ for sufficiently large $k$ since we picked $m = \Theta(\gamma k)$. All in all, this shows that $\norm{R - \widetilde R}_0 \leq \gamma k$ as claimed.

The remaining steps are easy to analyze: Since $p$ is a prime, the function $\pi(x) = (\sigma x + \tau) \bmod p$ is invertible on $[p]$ (assuming that $\sigma \neq 0$, which happens with high probability). As $\pi(A \conv B) = R$ and as $A \conv B$ is a vector of length $U < p$ it follows that $A \conv B = \pi^{-1}(R)$. In the same way, we obtain for $\widetilde C = \pi^{-1}(\widetilde R)$ that $\norm{A \conv B - \widetilde C}_0 \leq \gamma k$. In the final step we use that $q$ is large enough (larger than any entry in the convolution $A \conv B$), so we can safely cast $\widetilde C$ back to an integer vector.

Let us finally bound the running time of Part 3. The calls to \Fold{} and \Unfold{} take time $\Order(k \log^2(k/m) + k \log_k U + \polylog(q)) = \Order(k \log^2(1/\gamma) + \polylog(k, \norm A_\infty, \norm B_\infty))$. Computing $T = \Order(1/\gamma)$ convolutions of vectors of length $m = \Order(\gamma k)$ takes time at most
\begin{equation*}
    \Order\!\left(\frac1\gamma \mult D(\gamma k)\right) = \Order\!\left(\frac{\gamma k}\gamma \mult \frac{D(\gamma k)}{\gamma k}\right) = \Order\!\left(k \mult \frac{D(k)}{k}\right) = \Order(D(k)),
\end{equation*}
assuming that $D(n) / n$ is nondecreasing. Summing the two contributions yields the claimed running time.
\end{proof}

In combination, Lemmas~\ref{lem:set-query-part-1},~\ref{lem:set-query-part-2} and~\ref{lem:set-query-part-3} show that \autoref{alg:set-query} is correct and runs in the correct running time with probability at least $1 - \delta/2 - 1/\poly(k) \geq 1 - \delta$. This finishes the proof of \autoref{lem:tiny-univ-set-query}.

\subsection{Corrections for Randomized Dense Convolution} \label{sec:algorithm:sec:randomized}
In the previous subsection, we assumed that we have black-box access to a deterministic algorithm computing the dense convolution of two length-$n$ vectors in time $D(n)$. We will now prove that it suffices to assume that the black-box algorithm errs with constant probability, say $1/3$.

\begin{lemma} \label{lem:tiny-univ-set-query-randomized}
Let $\log k \leq 1/\gamma \leq \poly(k)$ and let $1/\delta \leq \poly(k)$. There is an algorithm for the \TinyUnivApproxSetQuery{} problem running in time
\begin{equation*}
    \Order(D_{1/3}(k) + k \log^2(1/\gamma) + k \log(1/\delta) + \polylog(k, \norm A_\infty, \norm B_\infty)).
\end{equation*}
\end{lemma}

The idea is simple: Every call to the randomized dense convolution algorithm is followed by a call to the verifier presented in \autoref{lem:dense-verification}. If a faulty output is detected, then we repeat the convolution (with fresh randomness) and test again. 

\begin{lemma}[Dense Verification] \label{lem:dense-verification}
Given three vectors $A, B, C$ of length $U$, there is a randomized algorithm running in time $\Order(U + \polylog(U, \norm A_\infty, \norm B_\infty))$, which checks whether $A \conv B = C$. The algorithm fails with probability at most $1 / \poly(U)$. 
\end{lemma}

The proof of \autoref{lem:dense-verification} is by a standard application of the classical Schwartz-Zippel lemma; in \autoref{sec:guess-k} we prove a more general statement about a sparse verifier. We also need the following tail bound on the sum of geometric random variables~\cite{Janson18}:

\begin{theorem}[{{\cite[Theorem~2.1]{Janson18}}}] \label{thm:geometric-tailbound}
Let $X_1, \dots, X_n$ be independent, identically distributed geometric random variables, and let $X = \sum_i X_i$. Then, for any $\lambda \geq 1$:
\begin{equation*}
    \Pr(X > \lambda \Ex(X)) \leq \exp(-n (\lambda - 1 - \ln\lambda)).
\end{equation*}
\end{theorem}

\begin{proof}[Proof of \autoref{lem:tiny-univ-set-query-randomized}]
The overall proof is exactly as for \autoref{lem:tiny-univ-set-query}, we merely substitute the black-box calls to the dense convolution algorithm. The only place where this algorithm is directly called is in the proof of \autoref{lem:set-query-part-3}, where we compute $T = \Order(1 / \gamma)$ convolutions of length $m = \Order(\gamma k)$. Each such call is replaced by a test-and-repeat loop using the verifier in \autoref{lem:dense-verification}. As the failure probability of the verifier is at most $1 / \poly(m) = 1/ \poly(k)$, we can afford a union bound and assume that the verifier never fails, i.e., we uphold the assumption that dense convolution succeeds.

It remains to bound the running time overhead. A single iteration of the test-and-repeat loop takes time $\Order(D_{1/3}(m) + m) = \Order(D_{1/3}(m))$. To bound the number of iterations $X = \sum_i X_i$, let~$X_i$ model the number of iterations caused by the $i$-th dense convolution call. Observe that $X_i$ is geometrically distributed with parameter $p = 2/3$ and thus $\Ex(X) = 3T/2$. By \autoref{thm:geometric-tailbound} with, say, $\lambda = 4$, it follows that $\Pr(X > 4 \Ex(X)) \leq \exp(-T) \leq \exp(-\Omega(1/\gamma))$. Using that $\log k \leq 1/\gamma$, the number of iterations is bounded by $6T$ with high probability $1 - 1/\poly(k)$ and therefore the total running time to answer all dense convolution queries is bounded by $\Order(T D_{1/3}(m)) = \Order(D_{1/3}(k))$.
\end{proof}

\section{Approximating the Support Set} \label{sec:approx-supp}
This section is devoted to finding a set $X$ which closely approximates $\supp(A\conv B)$. To that end, our goal is to solve the following problem, which is later applied with $Y = \supp(A)$ and $Z = \supp(B)$.

\probtinyunivapproxsupp*{}

\begin{lemma} \label{lem:tiny-univ-approx-supp}
There is an $\Order(k \log(1/\gamma) \log(\frac1{\gamma\delta}))$-time algorithm for the \prob{TinyUniv-ApproxSupp} problem.
\end{lemma}

A key ingredient to the algorithm is the following routine to approximately compute sumsets, which we shall refer to as Indyk's algorithm.

\begin{theorem}[Randomized Boolean Convolution~{{\cite{Indyk98}}}] \label{thm:indyk}
There exists an algorithm which takes as input two sets $Y, Z \subseteq [U]$, and in time $\Order(U)$ outputs a set $\mathcal O \subseteq Y + Z$, such that for all $x \in Y + Z$ we have $\Pr(x \in \mathcal O) \geq \frac{99}{100}$.
\end{theorem}

\begin{algorithm}[t]
\caption{$\TinyUnivApproxSupp(Y, Z, U, k)$} \label{alg:approx-supp}
\begin{algorithmic}[1]
\Input{Sets $Y, Z \subseteq [U]$ over a universe $U \leq k / \gamma$ such that $|Y + Z| \leq k$}
\Output{A set $X \subseteq [U]$ of size $\Order(k)$ such that $|(Y + Z) \setminus X| \leq \gamma k$}

\smallskip
\State Let $m = 40k$ and pick a prime $p \geq U$ 
\State Let $L = \ceil{\log(1/\gamma)}$
\State $X_L \gets \{0,1,\dots,\ceil{U/2^\ell}\}$
\For {$\ell \gets L - 1, \dots, 1, 0$}
	\State $Y_\ell \gets Y \bdiv 2^\ell$
	\State $Z_\ell \gets Z \bdiv 2^\ell$
	\State $M \gets 2X_{\ell+1} + \{0, 1, 2\}$ \label{line:expand}
	\RepeatTimes {$R = \Theta(\log(1/\gamma) + \log (1/\delta))$} \label{line:inner-loop}
		\State Randomly pick a linear hash function $h$ with parameters $p$ and $m$
	    	\State $\mathcal{O} \gets \text{output of Indyk's algorithm (\autoref{thm:indyk}) with input $h(Y_\ell), h(Z_\ell)$}$ \label{line:run-indyk}
		\For {$x \in M$}
			\If {$(h(0) + h(x) + o) \bmod m \in (\mathcal O \bmod m)$ for some $o \in \{-p, 0, p\}$}
				\State \label{line:vote} Give a vote to $x$
			\EndIf
		\EndFor
	\EndRepeatTimes
	\State $X_\ell \gets \text{all elements in $M$ that have gathered at least $3R / 4$ votes}$
\EndFor
\vspace*{-.8ex}
\State\Return $X = X_0$
\smallskip

\end{algorithmic}
\end{algorithm}

The algorithm claimed in \autoref{lem:tiny-univ-approx-supp} is given in \autoref{alg:approx-supp}. For the remainder of this section, we will analyze this algorithm in several steps. We shall call the iterations of the outer loop \emph{levels} and call an element $x$ a \emph{witness at level $\ell$} if $x \in Y_\ell + Z_\ell$. Otherwise, we say that $x$ is a non-witness. Fix a level $\ell$ and consider a single iteration of the inner loop (Lines~\ref{line:inner-loop}--\ref{line:vote}). The \emph{voting probability of~$x$ at level $\ell$} is the probability that $x$ is given a vote in \autoref{line:vote}. Recall that in every such iteration, we pick a random linear hash function $h : [U] \to [m]$ using fresh randomness. The following lemmas prove that witnesses have large voting probability and non-witnesses have small voting probability.

\begin{lemma}[Witnesses have Large Voting Probability] \label{lem:voting-prob-witness}
At any level $\ell$, the voting probability of a witness $x$ is at least $\frac{99}{100}$.
\end{lemma}
\begin{proof}
Recall that if $x$ is a witness at level $\ell$, then $x = y + z$ for some $y \in Y_\ell$ and $z \in Z_\ell$. By the almost-affinity of linear hashing (\autoref{lem:linear-hashing-basics}), it holds that $h(y) + h(z) = h(x) + h(0) + o \mod m$ for some offset $o \in \{-p, 0, p\}$. It follows that $(h(x) + h(0) + o) \bmod m$ is an element of the sumset $(h(Y_\ell) + h(Z_\ell)) \bmod m$. However, in order for $x$ to gain a vote, this condition must be true for the set $\mathcal O$ returned by Indyk's algorithm. By the guarantee of \autoref{thm:indyk}, $\mathcal O$~contains every element of $h(Y_\ell) + h(Z_\ell)$ with probability at least~$\frac{99}{100}$, which yields the claim. 
\end{proof}

\begin{lemma}[Non-Witnesses have Small Voting Probability] \label{lem:voting-prob-nonwitness}
At any level $\ell$, the voting probability of a non-witness $x$ is at most~$1/2$.
\end{lemma}
\begin{proof}
Given the fact that Indyk's algorithm never returns a false positive, it suffices to prove that none of the three values $(h(0) + h(x) + \{-p, 0, p\}) \bmod m$ is contained in the sumset $(h(Y_\ell) + h(Z_\ell)) \bmod m$, with sufficiently large probability. By the almost-affinity of $h$, we have
\begin{equation*}
	h(Y_\ell) + h(Z_\ell) \bmod m \subseteq (h(0) + h(Y_\ell + Z_\ell) + \{-p, 0, p\}) \bmod m.
\end{equation*}
So fix some offsets $o, o' \in \{-p, 0, p\}$ and some witness $x' \in Y_\ell + Z_\ell$. As $x$ is not a witness, we must have $x \neq x'$. It suffices to bound the following bound the probability:
\begin{equation*}
	\Pr(h(0) + h(x) + o = h(0) + h(x') + o' \bmod m) = \Pr(h(x) = (h(x') + o' - o) \bmod m) \leq{} \frac 4m,
\end{equation*}
where in the last step we applied the universality of $h$ (\autoref{lem:linear-hashing-basics}). By a union bound over the five possible values of $o' - o$ and over all witnesses $x'$, we conclude that the voting probability of $x$ is at most $20 |Y_\ell + Z_\ell| / m \leq 20 k / m \leq 1/2$.
\end{proof}

We are now ready to prove \autoref{lem:tiny-univ-approx-supp}. We shall do it in two steps: First we bound the running time and the number of false positives, i.e.\ $|X \setminus(Y+Z)|$, and second the number of false negatives, i.e.\ $|(Y+Z)\setminus X|$.

\begin{lemma}[Running Time of {{\autoref{alg:approx-supp}}}] \label{lem:approx-supp-running-time}
With probability $1 - \delta/2$, \autoref{alg:approx-supp} outputs a set~$X$ of size $\Order(k)$, and its running time is $\Order( k \log(1/\gamma) \log(\frac1{\gamma\delta}) )$.
\end{lemma}
\begin{proof}
Fix any level $\ell$. By \autoref{lem:voting-prob-nonwitness} we know that the voting probability of any non-witness~$x$ is at most $1/2$. Thus, by an application of Chernoff's bound, the probability that $x$ receives more than $3R/4$ votes over all $R = \Omega(\log L + \log(1/\delta))$ rounds is at most $2^{-\Omega(R)} \leq \delta / (12L)$ by appropriately choosing the constant in the definition of $R$ (in the upcoming \autoref{lem:approx-supp-correctness} we will see why $R$ is even slightly larger). By Markov's inequality, we obtain that with probability $1 - \delta / (2L)$ the number of non-witness elements in $M$ which will be inserted in $X_\ell$ is at most $|M| / 6 \leq 3 |X_{\ell+1}| / 6$. By a union bound over all levels, with probability $1 - \delta / 2$ we get that
\begin{equation*}
	|X_\ell| \leq k + \frac12 |X_{\ell+1}|,
\end{equation*}
for all $\ell \in [L]$. As initially $|X_L| \leq k$ it follows by induction that $|X_\ell| \leq (\sum_{i=0}^\infty 1/2^i) k = 2k$. In particular we have that $|X| = |X_0| = \Order(k)$, as claimed.

The total running time of the algorithm can be split into two parts:
\begin{enumerate*}[font=, label=(\roman*)]
\item the time spent on running Indyk's algorithm in \autoref{line:run-indyk}, and
\item the time needed to iterate over all elements $x \in M$ across all levels and assign them votes (\autoref{line:vote}).
\end{enumerate*}
The former is $\Order(m L R) = \Order(k L R)$ (recall that Indyk's algorithm runs for sets over the universe~$[m]$) and also the latter is
\begin{equation*}
	\sum_{\ell \in [L]} \Order(|X_\ell| R) = \sum_{\ell \in [L]} \Order(k R) = \Order(k L R).
\end{equation*}
Together, we obtain the desired bound on the running time $\Order(k L R) = \Order(k \log(1/\gamma) \log(\frac1{\gamma\delta}))$.
\end{proof}

\begin{lemma}[Correctness of \autoref{alg:approx-supp}] \label{lem:approx-supp-correctness}
With probability $1 - \delta / 2$, \autoref{alg:approx-supp} correctly outputs a set $X$ with $|(Y+Z)\setminus X| \leq \gamma k$.
\end{lemma}
\begin{proof}
Fix any $y \in Y, z \in Z$ and define $y_\ell = \floor{\frac y{2^\ell}}$, $z_\ell = \floor{\frac z{2^\ell}}$ and $x_\ell = y_\ell + z_\ell$. The first step is to prove that $x_\ell \in 2\{x_{\ell+1}\} + \{0, 1, 2\}$. Indeed, from the basic inequalities $2 \floor a \leq \floor{2a} \leq 2 \floor a + 1$, for all rationals~$a$, it follows directly that
\begin{align*}
	x_\ell - 2x_{\ell+1}
	={} \!\floor*{\frac y{2^\ell}} + \floor*{\frac z{2^\ell}} - 2 \floor*{\frac y{2^{\ell+1}}} - 2 \floor*{\frac z{2^{\ell+1}}}
	\leq 2,
\end{align*}
and in the same way $x_\ell - 2x_{\ell+1} \geq 0$.

Coming back to the algorithm, we claim that with probability $1 - \delta\gamma / 2$, $x = y + z$ will participate in $X$. It suffices to show that with the claimed probability, for all levels $\ell$ the element $x_\ell$ belongs to~$X_\ell$. Note that trivially $x_L \in X_L$. Fix a specific level $\ell$. Conditioning on $x_{\ell+1} \in X_{\ell+1}$, it will be the case that $x_\ell$ is inserted into $M = 2X_\ell + \{0, 1, 2\}$ in \autoref{line:expand}, by the fact that $x_\ell \in 2\{x_{\ell+1}\} + \{0, 1, 2\}$. Moreover, recall that $x_\ell$ is a witness at level $\ell$ and thus, by \autoref{lem:voting-prob-witness}, its voting probability is at least \raisebox{0pt}[0pt][0pt]{$\frac{99}{100}$}. Therefore it receives more than $3R/4$ votes and is inserted into $X_\ell$ with probability at least $1 - 2^{-\Omega(R)} \geq 1 - \delta\gamma / (2L)$. Taking a union bound over all levels we obtain that $x$ is contained in $X$ with probability $1 - \delta\gamma / 2$, and hence we can apply Markov's inequality to conclude that with probability $1 - \delta / 2$ it is the case that $|(Y+Z)\setminus X| \leq \gamma k$.
\end{proof}

This finishes the proof of \autoref{lem:tiny-univ-approx-supp}. Putting together the results from the previous section (\autoref{lem:tiny-univ-set-query-randomized}) and this section (\autoref{lem:tiny-univ-approx-supp} with $\gamma' = \order(\gamma^2)$), we have established an efficient algorithm to approximate convolutions in a tiny universe:

\begin{lemma} \label{lem:tiny-univ-approx-sparse-conv}
Let $\log k \leq 1/\gamma \leq \poly(k)$ and let $1/\delta \leq \poly(k)$. There is an algorithm for the \TinyUnivApproxSparseConv{} problem running in time
\begin{equation*}
	\Order(D_{1/3}(k) + k \log (1/ \gamma) \log(\tfrac1{\gamma\delta}) + \polylog(k, \norm A_\infty, \norm B_\infty)).
\end{equation*}
\end{lemma}

\section{Universe Reduction from Small to Tiny} \label{sec:small-to-tiny}

The goal of this section is to prove that approximating convolutions in a small universe (that is, a universe of size $U = \poly(k)$) reduces to approximating convolutions in a tiny universe.

\probsmallunivapproxsparseconv*{}

For the rest of the reduction we will set $\gamma = \delta$.

\begin{lemma} \label{lem:small-to-tiny}
Let $\log k \leq 1 / \delta = 1 / \gamma$. There is an algorithm for the \SmallUnivApproxSparseConv{} problem running in time $\Order(D_{1/3}(k) + k \log^2(1/\delta) + \polylog(k, \norm A_\infty, \norm B_\infty))$.
\end{lemma}

Our goal is to drastically reduce the universe size from $\poly(k)$ to $k / \delta^2$, while being granted to introduce up to $\delta k$ errors in the output. The idea is to use linear hashing to reduce the universe size. Then, to recover a vector in the original universe, we use a trick first applied by Huang~\cite{Huang19}. For a vector $A$, let~$\partial A$ denote the vector with $(\partial A)_i = i A_i$ (there is an analogy to derivatives of polynomials, hence the symbol). In essence we exploit the following familiar identity.

\begin{proposition}[Product Rule] \label{prop:product-rule}
Let $A, B$ be vectors. Then $\partial(A \conv B) = \partial A \conv B + A \conv \partial B$.
\end{proposition}
\begin{proof}
For any coordinate $x$:
\begin{equation*}
    (\partial A \conv B + A \conv \partial B)_x = \sum_{y + z = x} (\partial A)_y B_z + A_y (\partial B)_z = \sum_{y + z = x} (y + z) A_y B_z = (\partial (A \conv B))_x. \qedhere
\end{equation*}
\end{proof}

\begin{algorithm}[t]
\caption{$\SmallUnivApproxSparseConv(A, B, U, k)$} \label{alg:small-to-tiny}
\begin{algorithmic}[1]
\Input{Nonnegative vectors $A, B$ of length $U$, an integer $k$ such that $\norm{A \conv B}_0 \leq k$ and $U = \poly(k)$}
\Output{A vector $\widetilde C$ with $\norm{A \conv B - \widetilde C}_0 \leq \delta k$, with probability $1 - \delta$}

\smallskip
\State Let $p > U$ be a prime and let $m = \ceil{320 k / \delta^2}$
\State Randomly pick a linear hash function with parameters $p$ and $m$

\smallskip
\Statex \emph{(Approximate $V = h(A) \conv_m h(B)$)}
\State \raisebox{0pt}[0pt][0pt]{$\widetilde V^1 \gets \TinyUnivApproxSparseConv(h(A), h(B), m, k)$} with parameter $\delta / 6$
\State \raisebox{0pt}[0pt][0pt]{$\widetilde V \gets \widetilde V^1 \bmod m$}

\smallskip
\Statex \emph{(Approximate $W = h(\partial A) \conv_m h(B) + h(A) \conv_m h(\partial B)$)}
\State \raisebox{0pt}[0pt][0pt]{$\widetilde W^1 \gets \TinyUnivApproxSparseConv(h(\partial A), h(B), m, k)$} with parameter $\delta / 6$
\State \raisebox{0pt}[0pt][0pt]{$\widetilde W^2 \gets \TinyUnivApproxSparseConv(h(A), h(\partial B), m, k)$} with parameter $\delta / 6$
\State \raisebox{0pt}[0pt][0pt]{$\widetilde W \gets (\widetilde W^1 + \widetilde W^2) \bmod m$}

\medskip
\State $\widetilde C \gets 0$
\For {$i \in \supp(\widetilde V)$}
    \State $x \gets \widetilde W_i / \widetilde V_i$
    \If {$x$ is an integer}
        \State \raisebox{0pt}[0pt][0pt]{$\widetilde C_x \gets \widetilde C_x + \widetilde V_i$} \label{line:update-tilde-C}
    \EndIf
\EndFor \vspace{-.2ex}
\State\Return $\widetilde C$
\smallskip

\end{algorithmic}
\end{algorithm}

For the remainder of this section we analyze the procedure in \autoref{alg:small-to-tiny}. Let $h$ be a random linear hash function with parameters $p$ and $m$. We say that an index $x \in \supp(A \conv B)$ is \emph{isolated} if there is no other index $x' \in \supp(A \conv B)$ with $h(x') \in (h(x) + \{-2p, -p, 0, p, 2p\}) \bmod m$.

\begin{lemma}[Most Indices are Isolated] \label{lem:isolation-most}
With probability $1 - \delta / 2$, the number of non-isolated indices $x \in \supp(A \conv B)$ is at most $\delta k / 8$.
\end{lemma}
\begin{proof}
For any fixed integer $o$, the probability that $h(x') = h(x) + o \mod m$ is at most~$4/m$ by the universality of linear hashing (\autoref{lem:linear-hashing-basics}). By taking a union bound over the five values $o \in \{-2p, -p, 0, p, 2p\}$ and the $|\supp(A \conv B)| \leq k$ values of $x'$, the probability that $x$ is isolated is at least $1 - 20k/m$. As $m = 320 k / \delta^2$, any index $x \in \supp(A \conv B)$ is isolated with probability at least $1 - \delta^2 / 16$. The statement follows by an application of Markov's inequality.
\end{proof}

Recall that $V = h(A) \conv_m h(B)$ and $W = h(\partial A) \conv_m h(B) + h(A) \conv_m h(\partial B)$.

\begin{lemma}[Isolated Indices are Recovered] \label{lem:isolation-recovery}
Let $x$ be isolated. Then $\sum_i V_i = (A \conv B)_x$, where $i$ runs over $(h(0) + h(x) + \{-p, 0, p\}) \bmod m$. Furthermore, for any such $i$ it holds that $W_i = x V_i$.
\end{lemma}
\begin{proof}
Let $x \in \supp(A \conv B)$ be isolated and let $y \in \supp(A)$ and $z \in \supp(B)$ be such that $x \neq y + z$. We claim that $h(y) + h(z) \neq h(0) + h(x) + o \mod m$ for all $o \in \{-p, 0, p\}$. Assume the contrary, then by almost-affinity (\autoref{lem:linear-hashing-basics}) we have that $h(y) + h(z) = h(0) + h(y + z) + o' \mod m$ for some $o' \in \{-p, 0, p\}$ and thus $h(y + z) = h(x) + o - o' \mod m$. This is a contradiction as $y + z \in \supp(A \conv B)$ and $o - o' \in \{-2p, -p, 0, p, 2p\}$ but $x$ is assumed to be isolated.

Recall that $V = h(A) \conv_m h(B)$, and let $i \in (h(0) + h(x) + \{-p, 0, p\}) \bmod m$. For convenience, we write $\equiv$ to denote equality modulo $m$. From the previous paragraph it follows that
\begin{equation*}
    V_i = \sum_{\substack{y, z\\h(y) + h(z) \equiv i}} A_y B_z = \sum_{\substack{y + z = x\\h(y) + h(z) \equiv i}} A_y B_z.
\end{equation*}
By another application of almost-affinity it is immediate that $\sum_i V_i = \sum_{y + z = x} A_x B_y = (A \conv B)_x$. Moreover, we can express $W_i$ in a similar way: By repeating the previous argument twice, once with $\partial A$ in place of $A$ and once with $\partial B$ in place of $B$, we obtain that
\begin{equation*}
    W_i = \sum_{\substack{y + z = x\\h(y) + h(z) \equiv i}} (\partial A)_y B_z + \sum_{\substack{y + z = x\\h(y) + h(z) \equiv i}} A_y (\partial B)_z = \sum_{\substack{y + z = x\\h(y) + h(z) \equiv i}} (y + z) A_y B_z = x V_i. \qedhere
\end{equation*}
\end{proof}

\begin{lemma}[Correctness of \autoref{alg:small-to-tiny}]
With probability $1 - \delta$, \autoref{alg:small-to-tiny} correctly outputs a vector $\widetilde C$ with $\norm{A \conv B - \widetilde C}_0 \leq \Order(\delta k)$.
\end{lemma}
\begin{proof}
We call an iteration $i \in \supp(\widetilde V)$ \emph{good} if
\begin{enumerate*}[font=, label=(\roman*)]
\item $V_i = \widetilde V_i$,\label{lem:small-to-tiny-correctness:item:V}
\item $W_i = \widetilde W_i$, and\label{lem:small-to-tiny-correctness:item:W}
\item for some isolated element $x \in \supp(A \conv B)$ it holds that $i \in (h(0) + h(x) + \{-p, 0, p\}) \bmod m$.\label{lem:small-to-tiny-correctness:item:isolated}
\end{enumerate*}
Otherwise, $i$ is \emph{bad}. We start analyzing the algorithm with the unrealistic assumption that all iterations are good.

Focus on an arbitrary iteration $i$. By assumption~\ref*{lem:small-to-tiny-correctness:item:isolated} there exists some element $x \in \supp(A \conv B)$ such that $i \in (h(0) + h(x) + \{-p, 0, p\}) \bmod m$. Moreover, by \autoref{lem:isolation-recovery} and assumptions~\ref*{lem:small-to-tiny-correctness:item:V},~\ref*{lem:small-to-tiny-correctness:item:W} we have the algorithm correctly detects \raisebox{0pt}[0pt][0pt]{$x = \widetilde W_i / \widetilde V_i = W_i / V_i$}, and hence, over the course of the at most three iterations with $i \in (h(0) + h(x) + \{-p, 0, p\}) \bmod m$ we correctly assign $\widetilde C_x \gets \sum_i V_i = (A \conv B)_x$ in \autoref{line:update-tilde-C}. Under the unrealistic assumption it follows that \raisebox{0pt}[0pt][0pt]{$\widetilde C = A \conv B$} after all iterations.

We will now remove the unrealistic assumption. Clearly there is no hope of recovering the non-isolated elements, but \autoref{lem:isolation-most} proves that there are at most $\delta k / 8$ non-isolated elements with probability $1 - \delta / 2$. Any isolated element will be recovered if it happens to show up in a good iteration as shown in the previous paragraph. It suffices to prove that the number of bad iterations is at most $7 \delta k / 8$. Then, as any iteration modifies $\widetilde C$ in at most one position, it follows that $\norm{A \conv B - \widetilde C}_0 \leq \delta k$.

On the one hand, with probability $1 - \delta/2$, all three calls to the \TinyUnivApproxSparseConv{} algorithm succeed and we have that \raisebox{0pt}[0pt][0pt]{$\norm{V - \widetilde V}_0 \leq \delta k / 6$} and similarly \raisebox{0pt}[0pt][0pt]{$\norm{W - \widetilde W}_0 \leq \delta k / 3$}. So there can be at most $\delta k / 2$ iterations for which either assumption~\ref*{lem:small-to-tiny-correctness:item:V} or~\ref*{lem:small-to-tiny-correctness:item:W} fails. On the other hand, we assumed that there are at most $\delta k / 8$ non-isolated indices and any non-isolated index leads to at most three iterations for which~\ref*{lem:small-to-tiny-correctness:item:isolated} fails. It follows that in total there are at most $\delta k / 2 + 3 \delta / 8 = 7 \delta / 8$ bad iterations.
\end{proof}

It is easy to see that the running time of \autoref{alg:small-to-tiny} is dominated by the calls to \TinyUnivApproxSparseConv{} and thus bounded by $\Order(D_{1/3}(k) + k \log^2(1/\delta) + \polylog(k, \norm A_\infty, \norm B_\infty))$, see \autoref{lem:tiny-univ-approx-sparse-conv}. This completes the proof of \autoref{lem:small-to-tiny}.


\section{Error Correction} \label{sec:error_correction}
In the previous sections, the goal was design algorithms to \emph{approximate} convolutions. In this step, we show how to clean up the errors and turn the approximations into \emph{exact} convolutions. Formally, we give an algorithm for the following problem:

\probsmallunivsparseconv*{}

\begin{lemma} \label{lem:error-corr}
Let $\log^2 k \leq 1/\delta$. There is an algorithm for the \SmallUnivSparseConv{} problem running in time $\Order(D_{1/3}(k) + k \log^2(1/\delta) + \polylog(k, \norm A_\infty, \norm B_\infty))$.
\end{lemma}

This is the only part of the reduction for which we cannot use linear hashing, as the recovery loop crucially relies on certain cancellations to take place. The problem is that linear hashing is only almost -- not perfectly -- affine. Instead, we use a simpler hash function: Let $g(x) = x \bmod p$, where $p$ is a random prime in some specified range. The following basics are well-known and much simpler in comparison to linear hashing (cf.~\autoref{lem:linear-hashing-basics}).

\begin{lemma}[Hashing Modulo a Random Prime] \label{lem:prime-hashing-basics}
Let $g(x) = x \bmod p$ where $p$ is a random prime in the range $[m, 2m]$. Then the following properties hold:
\begin{description}
\item[Universality:] For distinct keys $x, y \in [U]$: $\Pr(g(x) = g(y)) \leq 2 \log(U) / m$.
\item[Affinity:] For arbitrary keys $x, y$: $g(x) + g(y) = g(x + y) \mod p$.
\end{description}
\end{lemma}
\begin{proof}
The affinity property is obvious, so focus on universality and fix two distinct keys $x, y \in [U]$. It holds that $g(x) = g(y)$ if and only if $p$ divides $x - y$. Since $|x - y| \leq U$, $x - y$ has at most $\log_m U$ distinct prime factors in the range $[m, 2m]$. On the other hand, by a quantitative version of the Prime Number Theorem~\cite[Corollary~3]{RosserS62}, there are at least $\frac{3m}{5 \ln m}$ primes in the range $[m, 2m]$ (for sufficiently large~$m$). Hence, the probability that $p$ divides $x - y$ is at most $5 \log_m(U) \ln(m) / (3m) \leq 2 \log(U) / m$.
\end{proof}

\begin{algorithm}[t]
\caption{$\SmallUnivSparseConv(A, B, U, k)$} \label{alg:small-univ-error-corr}
\begin{algorithmic}[1]
\Input{Nonnegative vectors $A, B$ of length $U$, an integer $k$ such that $\norm{A \conv B}_0 \leq k$ and $U = \poly(k)$}
\Output{$C = A \conv B$ with probability $1 - \delta$}

\smallskip
\State Let \raisebox{0pt}[0pt][0pt]{$m = \frac{8k \log(U)}{\log^2(k)}$}
\State $C^0 \gets \SmallUnivApproxSparseConv(A, B, U, k)$ with parameter $\delta/2$
\For {$\ell \gets 1, \dots, L = \Order(\log \log k)$}
    \RepeatTimes {$\ceil{2\log(2L/\delta) / 1.5^{\ell-1}}$} \label{line:inner-boosting-start}
        \State Randomly pick a prime $p \in [m, 2m]$ and let $g(x) = x \bmod p$
        \State $V \gets g(A) \conv_p g(B) - g(C^{\ell - 1})$ using FFT
        \State $W \gets g(\partial A) \conv_p g(B) + g(A) \conv_p g(\partial B) - g(\partial C^{\ell-1})$ using FFT \label{line:inner-boosting-end}
    \EndRepeatTimes \vspace*{-.7ex}
    \State Keep $g, V, W$ for which $\norm V_0$ is maximized \label{line:keep}
    \State $C^\ell \gets C^{\ell-1}$
    \For {$i \in \supp(V)$} \label{line:loop-Cell}
        \State $x \gets W_i / V_i$ \label{line:detect-x}
        \If {$x$ is an integer}
            \State $C^\ell_x \gets C^\ell_x + V_i$ \label{line:update-Cell}
        \EndIf
    \EndFor \vspace{-.2ex} 
\EndFor \vspace{-.4ex}
\State\Return $C = C^L$
\smallskip

\end{algorithmic}
\end{algorithm}

We analyze \autoref{alg:small-univ-error-corr}. We shall refer to iterations of the outer loop as \emph{levels $\ell$}. An element $x \in \supp(A \conv B - C^{\ell-1})$ is \emph{isolated at level~$\ell$} if there exists no $x' \in \supp(A \conv B - C^{\ell-1})$ with $x \neq x'$ and $g(x) = g(x')$, where $g$ is the function picked at the $\ell$-th level.

\begin{lemma}[Most Indices are Isolated] \label{lem:isolation-most-error-corr}
Let $\ell$ be any level. If
\begin{equation*}
    \norm{A \conv B - C^{\ell-1}}_0 \leq \frac{2^{-1.5^{\ell-1}} k}{\log^2(k)},
\end{equation*}
then with probability $1 - \delta / (2L)$, there will be at most
\begin{equation*}
    \frac{2^{-1.5^\ell} k}{2\log^2(k)}
\end{equation*}
non-isolated elements at level $\ell$.
\end{lemma}
\begin{proof}
We will prove the statement in three steps. First: \emph{A random hash function $g$ achieves that there are at most $2^{-1.5^\ell} k \log^{-2}(k) / 4$ non-isolated elements with probability at least $1 - \sqrt{2^{-1.5^{\ell-1}}}$.} Indeed, for any fixed $x \in \supp(A \conv B - C^{\ell-1})$ there are at most $\norm{A \conv B - C^{\ell-1}}_0$ other elements~$x'$ that $x$ could collide with. For fixed $x, x'$ the collision probability is at most
\begin{equation*}
    \Pr(g(x) = g(x')) \leq \frac{2 \log(U)}m \leq \frac{2 \log(U) \log^2(k)}{8 k \log(U)} = \frac{\log^2(k)}{4k}
\end{equation*}
by universality, see \autoref{lem:prime-hashing-basics}. By taking a union bound over all elements $x'$, we obtain that $x$ is non-isolated with probability at most $\norm{A \conv B - C^{\ell-1}}_0 \mult \log^2(k) / 4k$ and we thus expect at most $\norm{A \conv B - C^{\ell-1}}_0^2 \mult \log^2(k) / 4k$ non-isolated elements. By Markov's inequality, the probability that there are there are more than $2^{-1.5^\ell} k \log^{-2}(k) / 4$ non-isolated elements is at most
\begin{equation*}
    \frac{\norm{A \conv B - C^{\ell-1}}_0^2 \mult \frac{\log^2(k)}{4k}}{\frac{2^{-1.5^\ell} k}{4\log^2(k)}} \leq \frac{(2^{-1.5^{\ell-1}})^2}{2^{-1.5^\ell}} = \sqrt{2^{-1.5^{\ell-1}}}.
\end{equation*}

Second: \emph{With probability $1 - \delta / (2L)$, running $\ceil{2\log(2L/\delta) / 1.5^{\ell-1}}$ independent trials will result in at least one function $g$ under which there are at most $2^{-1.5^\ell} k \log^{-2}(k) / 4$ non-isolated elements} I.e.,~at some time during the execution of the inner loop in Lines~\ref{line:inner-boosting-start}--\ref{line:inner-boosting-end} we find a good hash function~$g$. Indeed, the failure probability is at most
\begin{equation*}
    \left(\sqrt{2^{-1.5^{\ell-1}}}\right)^{2 \log(2L/\delta) / 1.5^{\ell-1}} = \frac\delta{2L}.
\end{equation*}

Third: \emph{If there are at most $r$ non-isolated elements then $\norm V_0 \geq \norm{A \conv B - C^{\ell-1}}_0 - r$ and conversely, if $\norm V_0 \geq \norm{A \conv B - C^{\ell-1}}_0 - r$ then there can be at most $2r$ non-isolated elements.} Recall that by affinity (\autoref{lem:prime-hashing-basics}) the algorithm exactly computes $V = g(A \conv B - C^{\ell-1})$. As every isolated element $x$ is the unique element in its bucket $i = g(x)$ it follows directly that $\norm V_0 \geq \norm{A \conv B - C^{\ell-1}}_0 - r$ without accounting for the non-isolated elements. For the converse direction we note that there is a way of ``ignoring'' $r$ elements $x \in \supp(A \conv B - C^{\ell-1})$ such that all other elements become isolated. The number of non-isolated elements is thus at most $r$ (the ignored elements) plus $r$ (the number elements colliding with one of the ignored elements).

The statement follows by combining the second and third intermediate claims: By the second claim, the inner loop (Lines~\ref{line:inner-boosting-start}--\ref{line:inner-boosting-end}) will eventually discover some hash function $g$ under which we have at most $2^{-1.5^\ell} k \log^{-2}(k) / 4$ non-isolated elements and thus, by the third claim, $\norm V_0 \geq \norm{A \conv B - C^{\ell-1}}_0 - 2^{-1.5^\ell} k \log^{-2}(k) / 4$. As the algorithm selects the function which maximizes $\norm V_0$, the third claim proves that whatever function is kept in \autoref{line:keep} leads to at most $2^{-1.5^\ell} k \log^{-2}(k) / 2$ non-isolated elements.
\end{proof}

\begin{lemma}[Isolated Indices are Recovered] \label{lem:isolation-recovery-error-corr}
Denoting the umber of non-isolated elements at level~$\ell$ by~$r$, we have $\norm{A \conv B - C^\ell}_0 \leq 2r$.
\end{lemma}
\begin{proof}
Focus on arbitrary $\ell$, and assume that we already picked a hash function $g$ in Lines~\ref{line:inner-boosting-start}--\ref{line:keep}. By the affinity of $g$ it holds that $V = g(A \conv B - C^{\ell-1})$ and, by additionally using the product rule (\autoref{prop:product-rule}), $W = g(\partial(A \conv B - C^{\ell-1}))$.

Now focus on an arbitrary iteration $i \in \supp(V)$ of the second inner loop in Lines~\ref{line:loop-Cell}--\ref{line:update-Cell}. There must exist some $x \in \supp(A \conv B - C^{\ell-1})$ with $g(x) = i$. If $x$ is isolated then we will correctly set $C^\ell_x = (A \conv B)_x$. Indeed, from the isolation of $x$ it follows that $V_i = (A \conv B - C^{\ell-1})_x$ and $W_i = (\partial (A \conv B - C^{\ell-1}))_x = x V_i$. Thus $x$ is correctly detected in \autoref{line:detect-x} and in \autoref{line:update-Cell} we correctly assign $C^\ell_x \gets C^{\ell-1}_x + V_i = (A \conv B)_x$.

The previous paragraph shows that if at level $\ell$ all elements were isolated, we would compute $C^\ell = A \conv B$. We analyze how this guarantee is affected by the \emph{bad} iterations $i$ for which there exist several (non-isolated) elements $x \in \supp(A \conv B - C^{\ell-1})$ with $g(x) = i$. Clearly we cannot hope to correctly assign the $r$ entries $C^\ell_x$ for which $x$ is non-isolated. Additionally, there are at most~$r$ bad iterations, each of which potentially modifies $C^\ell$ in at most one position. We conclude that $\norm{A \conv B - C^\ell}_0 \leq 2r$.
\end{proof}

\begin{lemma}[Correctness of \autoref{alg:small-univ-error-corr}]
With probability $1 - \delta$, \autoref{alg:small-univ-error-corr} correctly outputs $C = A \conv B$.
\end{lemma}
\begin{proof}
We show that with probability $1 - \delta$ it holds that $\norm{A \conv B - C^\ell}_0 \leq 2^{-1.5^\ell} k \log^{-2}(k)$ for all levels~$\ell$. In particular, at the final level $L = \log_{1.5} \log k = \Order(\log \log k)$ we must have $\norm{A \conv B - C^\ell}_0 = 0$ and thus $A \conv B = C^\ell = C$.

The proof is by induction on $\ell \in [L + 1]$. For $\ell = 0$, the statement is true assuming that the \SmallUnivApproxSparseConv{} algorithm with parameter $\delta / 2 \leq \log^{-2}(k) / 2$ succeeds. For~\raisebox{0pt}[0pt][0pt]{$\ell > 1$}, we appeal to the previous lemmas: By the induction hypothesis we assume that $\norm{A \conv B - C^{\ell-1}}_0 \leq 2^{-1.5^{\ell-1}} k \log^{-2}(k)$. Hence, by \autoref{lem:isolation-most-error-corr}, the algorithm picks a hash function $g$ under which only $2^{-1.5^\ell} k \log^{-2}(k) / 2$ elements are non-isolated at level $\ell$. By \autoref{lem:isolation-recovery-error-corr} it follows that $\norm{A \conv B - C^\ell}_0 \leq 2^{-1.5^\ell} k \log^{-2}(k)$, which is exactly what we intended to show.

Let us analyze the error probability: For $\ell = 0$, the error probability is $\delta / 2$. For any other level (there are at most $L$ such), the error probability is $1 - \delta/(2L)$ by \autoref{lem:isolation-most-error-corr}. Taking a union bound over these contributions yields the claimed error probability of $1 - \delta$.
\end{proof}

\begin{lemma}[Running Time of \autoref{alg:small-univ-error-corr}]
The running time of \autoref{alg:small-univ-error-corr} is bounded by $\Order(D_{1/3}(k) + k \log^2(1/\delta) + \polylog(k, \norm A_\infty, \norm B_\infty))$.
\end{lemma}
\begin{proof}
We invoke \SmallUnivApproxSparseConv{} a single time with parameter $\delta/2$ which takes time exactly as claimed, see \autoref{lem:small-to-tiny}. After that, the running is $\Order(k \log (1/\delta))$ mostly due to the inner loop in Lines~\ref{line:inner-boosting-start}--\ref{line:inner-boosting-end}: A single execution of the loop body takes time $\Order(k)$ for hashing the six vectors $A, \partial A, B, \partial B, C^{\ell-1}, \partial C^{\ell-1}$ and for computing three convolutions of vectors of length $m = \Order(k / \log k)$ using FFT. It remains to bound the number of iterations:
\begin{equation*}
    \sum_{\ell = 1}^L \ceil*{\frac{2 \log(2L / \delta)}{1.5^{\ell - 1}}} \leq L + \sum_{\ell = 1}^L \frac{2 \log(2L / \delta)}{1.5^{\ell - 1}} = \Order(L + \log(L/\delta)) = \Order(\log(1/\delta)). \qedhere
\end{equation*}
\end{proof}

This finishes the proof of \autoref{lem:error-corr}.

\section{Estimating \texorpdfstring{$k$}{k}} \label{sec:guess-k}

In this section we remove the assumption that an estimate $k \geq \norm{A \conv B}_0$ is given as part of the input. Let us redefine the meaning of $k$ as $k = \norm{A \conv B}_0$ and refer to the estimate as $k^* \geq k$.

\begin{lemma} \label{lem:guess-k}
Let $\log^2 k \leq 1/\delta$. There is an algorithm for the restricted \SmallUnivSparseConv{} problem which does not expect a bound $k^* \geq \norm{A \conv B}_0$ as part of the input, running in time $\Order(D_{1/3}(k) + k \log^2(1/\delta) + \polylog(k, \norm A_\infty, \norm B_\infty))$.
\end{lemma}

The idea is to use exponential search to guess some $k^* \geq \norm{A \conv B}_0$. We need the following subroutine:

\begin{lemma}[Sparse Verification] \label{lem:sparse-verification}
Given three vectors $A, B, C$ of length $U = \poly(k)$ and sparsity at most $k$, there is a randomized algorithm running in time $\Order(k + \polylog(k, \norm A_\infty, \norm B_\infty))$, which checks whether $A \conv B = C$. The algorithm fails with probability at most $1/\poly(k)$.
\end{lemma}

The idea is standard: We can view $A$, $B$ and $C$ as polynomials via $A = \sum_i A_i x^i$ and similarly for $B$ and $C$. In that viewpoint, the role of the convolution operator is taken by polynomial multiplication, i.e.\ it suffices to check whether $A B = C$. This is a polynomial identity testing problem which can be classically solved by the Schwartz-Zippel lemma.

\begin{proof}
First check whether $\norm C_\infty > k \norm A_\infty \norm B_\infty$ and reject in this case. Otherwise compute a prime $p > k U + k \norm A_\infty \norm B_\infty$. We view $A$, $B$ and $C$ as polynomials over~$\Int_p$, by interpreting $A = \sum_i A_i x^i$ and similarly for $B$ and $C$. Next, sample a random point $x \in \Int_p$. We use the bulk exponentiation algorithm (\autoref{lem:bulk-exponentiation}) to precompute all relevant powers $x^i$ and then evaluate $A(x)$, $B(x)$ and $C(x)$ at $x$. If $A(x) B(x) = C(x)$, then we accept (confirming that $A \conv B = C$), otherwise we reject.

If $A \conv B = C$, then this algorithm is always correct. So suppose that $A \conv B \neq C$, and let $D = A B - C$. Over the integers it is clear that $D$ is not the zero polynomial, and since $p$ is large enough it also holds that $D$ is nonzero over $\Int_p$. The algorithm essentially evaluates $D$ at a random point $x \in \Int_p$ and accepts if and only if $D(x) = 0$. The error event is that $D(x) = 0$ despite $D$ being nonzero as a polynomial. Recall that $D$ has degree at most $U$, so it has at most $U$ zeros. Therefore, the probability of hitting a zero is at most $U/p < 1/k$.

Finally, we analyze the running time. Computing $p$ takes time $\polylog(k, \norm A_\infty, \norm B_\infty)$. Precomputing the powers of $x$ takes time $\Order(k \log_k U) = \Order(k)$ by \autoref{lem:bulk-exponentiation}, and also evaluating~$A$,~$B$ and~$C$ at~$x$ takes time $\Order(k)$. Note that all arithmetic operations carried out in these steps are over $\Int_p$ and since a single element of $\Int_p$ can be written down using a constant number of machine words, each ring operation takes constant time. As claimed, the total time is $\Order(k + \polylog(k, \norm A_\infty, \norm B_\infty))$.
\end{proof}

\begin{proof}[Proof of \autoref{lem:guess-k}]
The algorithm is simple: Let $k_0 = \norm A_0 + \norm B_0$ and loop through all estimates $k^* \gets 2^0 k_0, 2^1 k_0, 2^2 k_0, \dots$. For every such $k^*$, run the \prob{SmallUniv-SparseConv} algorithm with estimate $k^*$ to compute a vector $C$, followed by a call to the sparse verifier (\autoref{lem:sparse-verification}) which checks whether indeed $A \conv B = C$.

We write $k = \norm{A \conv B}_0$ for the actual sparsity of the convolution. After $\ceil{\log(k / k_0)}$ iterations we cross the threshold $k \leq k^*$. At this point, the \prob{SmallUniv-SparseConv} algorithm is guaranteed to be correct with probability, say, $\delta / 2$. We claim that also the verifier is correct in every iteration: By the trivial bound $k \leq \norm A_0 \mult \norm B_0 \leq k_0^2$ it follows that the verifier fails with probability at most $1/\poly(k^*) \leq 1/\poly(k_0) \leq 1/\poly(k)$. By a union bound over the $\Order(\log k)$ iterations it follows that with probability $1/\poly(k)$ the verifier never errs. The total error probability is thus $1 - \delta/2 - 1/\poly(k) \leq 1 - \delta$.

It remains to bound the running time. We can assume that each call to our sparse convolution algorithm runs in time $\Order(D_{1/3}(k^*) + k^* \log^2(1/\delta) + \polylog(k^*, \norm A_\infty, \norm B_\infty))$ even if we provide a wrong estimate $k^*$. This assumption can be guaranteed by counting the number of computation steps and aborting after the algorithm exceeds this time budget. Therefore, the total running time is
\begin{align*}
    &\sum_{i = 0}^{\ceil{\log(k / k_0)}} \Order(D_{1/3}(k / 2^i) + k / 2^i \log^2(1/\delta) + \polylog(k / 2^i, \norm A_\infty, \norm B_\infty)) \\
    \leq{} &\Order\!\left(\sum_{i = 0}^\infty D_{1/3}(k / 2^i)\right) + \Order(k \log^2 (1/\delta) + \polylog(k, \norm A_\infty, \norm B_\infty)).
\end{align*}
Recall that we assume that $D_{1/3}(n) / n$ is a nondecreasing function, hence the first term can be bounded by
\begin{equation*}
    \Order\!\left(\sum_{i = 0}^\infty \frac{k}{2^i} \frac{D_{1/3}(k / 2^i)}{k / 2^i}\right) = \Order\!\left(\sum_{i = 0}^\infty \frac{k}{2^i} \frac{D_{1/3}(k)}{k}\right) = \Order(D_{1/3}(k)). \qedhere
\end{equation*}
\end{proof}

\section{Universe Reduction from Large to Small} \label{sec:large-to-small}

The final step in our chain of reductions is to reduce from an arbitrarily large universe to a small universe (that is, a universe of size $U = \poly(k)$). We thereby solve the \SparseConv{} problem and complete the proof of \autoref{thm:core}.

\probsparseconv*{}

\begin{lemma} \label{lem:large-univ-sparse-conv}
Let $\log^2 k \leq 1/\delta$. There is an algorithm for the \SparseConv{} problem running in time $\Order(D_{1/3}(k) + k \log^2(1/\delta) + \polylog(U, \norm A_\infty, \norm B_\infty))$.
\end{lemma}

We will prove \autoref{lem:large-univ-sparse-conv} by analyzing \autoref{alg:large-to-small}, which in essence is a simpler version of \autoref{alg:small-to-tiny}. For that reason, we will be brief in this section. Recall that we cannot assume to have an estimated upper bound $k^*$ on $k = \norm{A \conv B}_0$.

As in \autoref{sec:small-to-tiny}, we call $x \in \supp(A \conv B)$ \emph{isolated} if there exists no other index $x' \in \supp(A \conv B)$ with $h(x') \in (h(x) + \{-2p, -p, 0, p, 2p\}) \bmod m$.

\begin{algorithm}[t]
\caption{$\alg{SparseConv}(A, B, U)$} \label{alg:large-to-small}
\begin{algorithmic}[1]
\Input{Nonnegative vectors $A, B$ of length $U$}
\Output{$C = A \conv B$, with probability $1 - \delta$}

\smallskip
\State Let $m = (\norm A_0 \mult \norm B_0)^3$ and let $p > U m$ be a prime
\State Randomly pick a linear hash function with parameters $p$ and $m$

\smallskip
\Statex \emph{(Compute $V = h(A) \conv_m h(B)$)}
\State $V^1 \gets \prob{SparseUniv-SparseConv}(h(A), h(B), m)$ with parameter $\delta / 6$
\State $V \gets V^1 \bmod m$

\smallskip
\Statex \emph{(Compute $W = h(\partial A) \conv_m h(B) + h(A) \conv_m h(\partial B)$)}
\State $W^1 \gets \prob{SmallUniv-SparseConv}(h(\partial A), h(B), m)$ with parameter $\delta / 6$
\State $W^2 \gets \prob{SmallUniv-SparseConv}(h(A), h(\partial B), m)$ with parameter $\delta / 6$
\State $W \gets (W^1 + W^2) \bmod m$

\medskip
\State $C \gets (0, \dots, 0) \in \Int^U$
\For {$i \in \supp(V)$} \label{line:loop-C}
    \State $x \gets W_i / V_i$ \label{line:divide-WV}
    \If {$x$ is an integer}
        \State $C_x \gets C_x + V_i$ \label{line:update-C}
    \EndIf
\EndFor \vspace{-.4ex}
\State\Return $C$
\smallskip

\end{algorithmic}
\end{algorithm}

\begin{lemma}[All Indices are Isolated] \label{lem:isolation-all}
With probability $1 - 1 / \poly(k)$, all indices $x \in \supp(A \conv B)$ are isolated.
\end{lemma}
\begin{proof}
The probability that $h(x') = h(x) + o \mod m$ is at most~$4/m$, for any distinct keys $x, x'$ and any fixed integer $o$, by the universality of linear hashing (\autoref{lem:linear-hashing-basics}). By taking a union bound over the five values $o \in \{-2p, -p, 0, p, 2p\}$ and the $k^2$ values of $(x, x')$, the probability that all indices~$x$ are isolated is at least $1 - 20k^2/m$. The statement follows since $m = (\norm A_0 \mult \norm B_0)^3 \geq k^3$.
\end{proof}

\begin{proof}[Proof of \autoref{lem:large-univ-sparse-conv}]
To use the \SmallUnivSparseConv{} algorithm, we only have to guarantee that the hashed vectors have length at most $\poly(k)$, which is true by $m = (\norm A_0 \mult \norm B_0)^3 \leq k^6$. Therefore, with probability $1 - \delta/2$ it holds that $V$ and $W$ are correctly computed. And, by \autoref{lem:isolation-all}, with probability $1 - 1/\poly(k)$ we have that all indices~$x$ are isolated. Both events happen simultaneously with probability at least $1 - \delta$, so for the rest for the proof we condition on both these events. By exactly the same proof as \autoref{lem:isolation-recovery} we get that
\begin{equation*}
    (A \conv B)_x = \sum_i V_i,
\end{equation*}
where $i$ runs over $(h(0) + h(x) + \{-p, 0, p\}) \bmod m$, and, for any such $i$ it holds that $W_i = x V_i$. In particular, in \autoref{line:divide-WV} we correctly identify $x = W_i / V_i$ and thus correctly assign $C_x = \sum_i V_i$ over the course of the at most three iterations $i$.

As the recovery loop in Lines~\ref{line:loop-C}--\ref{line:update-C} takes time $\Order(k)$, the total running time is dominated by the convolutions in a small universe taking time $\Order(D_{1/3}(k) + k \log^2(1/\delta) + \polylog(U, \norm A_\infty, \norm B_\infty))$, as shown in \autoref{lem:error-corr}. Here, in contrast to before, we cannot replace $U$ by $k$ in the additive term $\polylog(U, \norm A_\infty, \norm B_\infty)$, since the entries of $\partial A$ are as large as $U \norm A_\infty$.
\end{proof}

\section{Concentration Bounds for Linear Hashing} \label{sec:hashing}
In this section we sharpen the best-known concentration bounds for the most classic textbook hash function
\begin{equation*} \label{eq:linear-hashing}
    h(x) = ((\sigma x + \tau) \bmod p) \bmod m,
\end{equation*}
for a certain range of parameters. Here, $p$ is some (fixed) prime, $m \leq p$ is the (fixed) number of buckets and $\sigma, \tau \in [p]$ are chosen uniformly and independently at random. We say that $h$ is a \emph{linear hash function} with parameters $p$ and $m$.

Our main goal is to prove the following \autoref{thm:threewise-independence}, which is essential for the analysis of our sparse convolution algorithm and which we believe to be of independent interest. The result is based on the machinery established by Knudsen~\cite{Knudsen16}. In that work~\cite{Knudsen16}, Knudsen gives the following improved concentration bounds for $h$, similarly (but also crucially different from) the ones achieved by three-wise independent hash functions. For completeness we repeat some of Knudsen's proofs.

\begin{theorem}[Close to Three-Wise Independence {{\cite[Theorem 5]{Knudsen16}}}] \label{thm:knudsen-threewise-independence}
Let $X \subseteq [U]$ be a set of~$k$ keys. Randomly pick a linear hash function $h$ with parameters $p > 4U^2$ and $m \leq U$, fix a key $x \not\in X$ and buckets $a, b \in [m]$. Moreover, let $y, z \in X$ be chosen independently and uniformly at random. Then:
\begin{equation*}
    \Pr(h(y) = h(z) = b \mid h(x) = a) \leq \frac1{m^2} + \frac{2^{\Order(\sqrt{\log k \log\log k})}}{m k}.
\end{equation*}
\end{theorem}

Unfortunately, as we will prove later, there is no hope of improving the $2^{\Order(\sqrt{\log k \log\log k})}$-overhead by much in general. Fortunately, we prove that there is a loop hole: In small universes $U$, tighter bounds are possible:

\begin{theorem}[Closer to Three-Wise Independence in Small Universes] \label{thm:threewise-independence}
Let $X \subseteq [U]$ be a set of~$k$ keys. Randomly pick a linear hash function $h$ with parameters $p > 4U^2$ and $m \leq U$, fix a key $x \not\in X$ and buckets $a, b \in [m]$. Moreover, let $y, z$ be chosen independently and uniformly at random. Then:
\begin{equation*}
    \Pr(h(y) = h(z) = b \mid h(x) = a) \leq \frac1{m^2} + \Order\!\left(\frac{U \log U}{m k^2}\right)\!.
\end{equation*}
\end{theorem}

Our result improves upon \autoref{thm:knudsen-threewise-independence} when $U \leq k \mult 2^{\order(\sqrt{\log k \log\log k})}$. For $U \leq k \polylog k$ and $m \approx k$ (which is the relevant case for us), \autoref{thm:threewise-independence} provides a bound which is worse only by a poly-logarithmic factor in comparison to a truly three-wise independent hash function. (Indeed, for a truly three-wise independent hash function, the above probability is exactly $1/m^2$.) We apply \autoref{thm:threewise-independence} by means of the following corollary.

\coroverfullbuckets{}

We remark that the same concentration bounds can be obtained for the related family of hash functions $x \mapsto \big\lfloor\frac{((\sigma x + \tau) \bmod p) m}p\big\rfloor$.

The remainder of this section is structured as follows: In \autoref{sec:heights} we recap some basic definitions from~\cite{Knudsen16} and prove -- as the key step -- a certain number-theoretic bound. In \autoref{sec:hashing-proof} we then give proofs of \autoref{thm:threewise-independence} and \autoref{cor:overfull-buckets} closely following Knudsen's proof outline. Finally, in \autoref{sec:hashing-lower-bound} we prove that the concentration bound in \autoref{thm:knudsen-threewise-independence} is almost optimal by giving an almost matching lower bound.

\subsection{Heights} \label{sec:heights}
\def\residue#1{\boldsymbol{#1}}
We start with some definitions. Let $p$ be a prime. There is a natural way to embed $\Int$ into $\Int_p$ via $\iota(x) = x \bmod p$. As we often switch from $\Int$ to $\Int_p$, we introduce some shorthand notation: For $x \in \Int$, we use boldface symbols to abbreviate $\residue x = \iota(x)$.

The central concept of this proof is an arithmetic measure called the \emph{height $H(x)$} of a nonzero rational number $x \in \Rat$, which is defined as $\max(|a|, |b|)$ if $x$ can be written as $\frac ab$ for coprime integers~$a, b$. We also define a similar height measure for $\Int_p$: The height $H_p(\residue x)$ of $\residue x \in \Int_p^\times$ is defined as
\begin{equation*}
    H_p(\residue x) = \min\big\{ \max(\,|a|, |b|\,) : a, b \in \Int, \residue x = \residue a \residue b^{-1} \big\}.
\end{equation*}

\begin{lemma}[Equivalence of Heights {{\cite[Lemma 2]{Knudsen16}}}] \label{lem:heights-equivalence}
\hspace{0cm}
\begin{itemize}
\item Let $x, y$ be nonzero integers with $|x|, |y| < \sqrt{p/2}$. Then $H(\frac xy) = H_p(\residue x \residue y^{-1})$.
\item For all $x$, $H_p(\residue x) < \sqrt p$.
\end{itemize}
\end{lemma}
\begin{proof}
We start with the first item. It is clear that \raisebox{0pt}[0pt][0pt]{$H_p(\residue x \residue y^{-1}) \leq H(\frac xy)$} as whenever we can write $\frac xy = \frac ab$ for integers $a, b$, then we also have \raisebox{0pt}[0pt][0pt]{$\residue x \residue y^{-1} = \residue a \residue b^{-1}$}.

Next, we prove that \raisebox{0pt}[0pt][0pt]{$H(\frac xy) \leq H_p(\residue x \residue y^{-1})$}. Recall that we have \raisebox{0pt}[0pt][0pt]{$H_p(\residue x \residue y^{-1}) \leq \max(|x|, |y|) < \sqrt{p/2}$}. Suppose that $\residue x \residue y^{-1} = \residue a \residue b^{-1}$ for some $\residue a, \residue b \in \Int_p^\times$ such that \raisebox{0pt}[0pt][0pt]{$\max(|a|, |b|) = H_p(\residue x \residue y^{-1}) < \sqrt{p / 2}$}. Then $xb - ya$ must be an integer divisible by $p$, but \raisebox{0pt}[0pt][0pt]{$|a|, |b|, |x|, |y| < \sqrt{p / 2}$}. It follows that $|xb - ya| < p$ and thus $xb - ya = 0$. Finally, $\frac xy = \frac ab$ and $H(\frac xy) \leq \max(|a|, |b|) = H_p(\residue x \residue y^{-1})$.

It remains to prove the second item. Let $r = \floor{\sqrt p}$ and let $\residue S = \{\residue j \residue x : j \in \{0, \dots, r\}\}$. Observe that $\residue S$ contains $r + 1$ distinct elements. Let $0 = s_0 < \dots < s_r < p$ denote the unique integers such that $\residue S = \{ \residue s_0, \dots, \residue s_r \}$, and let $s_{r+1} = p$. Then we have $\sum_{i=0}^r s_{i+1} - s_i = p$, and thus there exists some index $i$ with $s_{i+1} - s_i \leq \frac p{r+1} < \sqrt p$. By the definition of $\residue S$ we can write $\residue x = (\residue s_{i+1} - \residue s_i) \residue j^{-1}$ for some $j \in \{1, \dots, r\}$ and hence $H_p(\residue x) \leq \max(s_{i+1} - s_i, r) < \sqrt p$.
\end{proof}

The next lemma constitutes the heart of our concentration bound. Knudsen~\cite[Corollary 4]{Knudsen16} shows that for any set $X$ the sum $\sum_{x, y \in X} 1/H(x/y)$ can be bounded by $k 2^{\Order(\sqrt{\log k \log \log k})}$ regardless of the universe size~$U$; in our setting (where $U$ is as small as $k \polylog k$) the following bound is significantly sharper.

\begin{lemma}[Sum of Inverse Heights] \label{lem:heights-small-universe}
Let $X \subseteq \{-U, \dots, U\}$ be a set of nonzero integers. Then:
\begin{equation*}
    \sum_{x, y \in X} \frac1{H(\frac xy)} \leq \Order(U \log U).
\end{equation*}
\end{lemma}
\begin{proof}
Let us assume that $X$ contains only positive integers; in the general case the sum can be at most four times larger since $H(\frac xy) = H(-\frac xy)$. We start with the following simple observation: If~$x, y$ are positive integers, then $H(\frac xy) = \max(x, y) / \gcd(x, y)$. Therefore, the goal is to bound
\begin{equation*}
    \sum_{x, y \in X} \frac1{H(\frac xy)} = \sum_{x, y \in X} \frac{\gcd(x, y)}{\max(x, y)} \leq 2 \sum_{x \in X} \frac1x \sum_{\substack{y \in X\\x \geq y}} \gcd(x, y).
\end{equation*}
Fix $x$ and focus on the sum $\sum_y \gcd(x, y)$. Let $g$ be a divisor of $x$. Then there can be at most $x / g$ values $y \leq x$ which are divisible by $g$. It follows that $|\{ y \leq x : \gcd(x, y) = g \}| \leq x / g$. Thus:
\begin{equation*}
    \sum_{\substack{y \in X \\ x \geq y}} \gcd(x, y) \leq \sum_{g | x} g \mult \frac xg = x \mult d(x), 
\end{equation*}
where $d(x)$ denotes the number of divisors of $x$. Combining these previous equations, we obtain
\begin{equation*}
    \sum_{x, y \in X} \frac1{H(\frac xy)} \leq 2 \sum_{x \in X} d(x) \leq 2 \sum_{x \in \{1, \dots, U\}} d(x).
\end{equation*}
To bound the right hand side by $\Order(U \log U)$, it suffices to check that the average number in $[U]$ has $\Order(\log U)$ divisors. More precisely: Any integer $g \in [U]$ divides at most $U / g$ elements in $[U]$ and therefore $\sum_x d(x) \leq \sum_g U / g = \Order(U \log U)$.
\end{proof}

\subsection{Proof of \autoref{thm:threewise-independence}} \label{sec:hashing-proof}
We need some technical lemmas proved in~\cite{Knudsen16}; for the sake of completeness we also give short proofs. Let us call a set $\residue I = \{ \residue a + \residue i \residue b : i \in [r] \} \subseteq \Int_p$ an \emph{arithmetic progression}, and if $\residue b = \residue 1$ then we call~$\residue I$ an \emph{interval}. For a set $\residue X \subseteq \Int_p$, we define the \emph{discrepancy} as
\begin{equation*}
    \disc(\residue X) = \max_{\text{\normalfont$\residue I$ interval}} \left|\, |\residue X \cap \residue I| - \frac{|\residue X| \, |\residue I|}p \,\right|.
\end{equation*}

\begin{lemma}[{{\cite[Lemma 3]{Knudsen16}}}] \label{lem:discrepancy-inverses}
Let $x, y$ be coprime integers with \raisebox{0pt}[0pt][0pt]{$|x|, |y| < \sqrt p$} and let $\residue X$ be an interval of length $|x|$. Then $\disc(\residue y \residue x^{-1} \residue X) \leq 2$.
\end{lemma}
\begin{proof}
For simplicity assume that $x, y$ are positive (the other cases are symmetric) and also assume that $\residue X = \{ \residue i : i \in [x] \}$ (which is enough, since the discrepancy is invariant under shifts). We first show that $\residue x^{-1} \residue X$ is evenly distributed in the following strong sense: $\residue x^{-1} \residue X = \residue Y$, where $\residue Y = \{ \iota(\ceil{j p / x}) : j \in [x] \}$. Since $\residue x^{-1} \residue X$ and~$\residue Y$ are finite sets of the same size, it suffices to prove the inclusion $\residue x^{-1} \residue X \subseteq \residue Y$. So fix any~$i \in [x]$; we show that $\residue i \residue x^{-1} \in \residue Y$. Let $j \in [x]$ be the unique integer such that $x$ divides $jp + i$. Then $\iota((jp + i)/x) = \residue i \residue x^{-1}$ and $\ceil{jp / x} = (jp + i) / x$ and hence $\residue i \residue x^{-1} \in \residue Y$.

The next step is to show that also $\residue y \residue x^{-1} \residue X$ is distributed evenly. From the previous paragraph we know that $\residue y \residue x^{-1} \residue X = \{ \iota(jp y / x + \epsilon_j y) : j \in [x]\}$ for some rational values $\epsilon_j = \ceil{jp / x} - jp / x < 1$. The key insight is that since $x$ and $y$ are coprime integers, the set $[x]$ is invariant under the dilation with $y$, that is, $\{ jy \bmod x : j \in [x] \} = [x]$. It follows that $\residue y \residue x^{-1} \residue X = \{ \iota(jp / x + \delta_j) : j \in [x]\}$ for some $\delta_j$'s with~$0 \leq \delta_j < y < p / x$.

We point out how to conclude that $\disc(\residue y \residue x^{-1} \residue X) \leq 2$. First, it is obvious that all intervals of the form $\{ \residue i : \ceil{jp/x} \leq i < \ceil{(j+1)p/x} \}$ intersect $\residue y \residue x^{-1} \residue X$ in exactly one point. As every interval~$\residue I$ can be decomposed into several such segments plus two smaller parts of size less than $p/x$ at the beginning and the end, respectively, a simple calculation confirms that $\disc(\residue y \residue x^{-1} \residue X) \leq 2$.
\end{proof}

\begin{lemma}[{{\cite[Lemma 4]{Knudsen16}}}] \label{lem:progression-independence}
Let $\residue x, \residue y \in \Int_p$ and let $\residue I \subseteq \Int_p$ be an arithmetic progression. Then, for $\residue\sigma \in \Int_p$ chosen uniformly at random:
\begin{equation*}
    \Pr((\residue\sigma \residue x, \residue\sigma \residue y) \in \residue I^2) \leq \frac{|\residue I|^2}{p^2} + \Order\!\left(\frac1{H_p(\residue x \residue y^{-1})} \mult \frac{|\residue I|}p + \frac1p\right)\!,
\end{equation*}
\end{lemma}
\begin{proof}
First, observe that $\residue I = \residue z \residue J$ for some interval $\residue J$ and some nonzero $\residue z$. As $(\residue\sigma \residue x, \residue\sigma \residue y) \in \residue I^2$ holds if and only if $(\residue\sigma \residue x \residue z^{-1}, \residue\sigma \residue y \residue z^{-1}) \in \residue J^2$, we may replace $\residue x, \residue y$ by $\residue x \residue z^{-1}, \residue y \residue z^{-1}$ and assume that $\residue I$ is an interval. Note that $H_p(\residue x \residue y^{-1})$ is invariant under this exchange. We may further scale and possibly exchange~$\residue x$ and~$\residue y$ to ensure that $H_p(\residue x \residue y^{-1}) = x \geq |y|$, where $x, y$ are integers such that $\iota(x) = \residue x$ and $\iota(y) = \residue y$ with $x$ positive and $x, y$ coprime.

Pick $\residue \sigma \in \Int_p$ uniformly at random, and let $\residue S = \{ \residue\sigma + \residue i \residue x^{-1} : i \in [x]\}$. Instead of bounding the probability $\Pr((\residue\sigma \residue x, \residue\sigma \residue y) \in \residue I^2)$ directly, by linearity of expectation we may instead bound
\begin{align*}
    &\Pr((\residue\sigma \residue x, \residue\sigma \residue y) \in \residue I^2) \\
    ={} &\frac1x \Ex\left( \sum_{\residue s \in \residue S} [(\residue s \residue x, \residue s \residue y) \in \residue I^2] \right) \\
    \leq{} &\frac1x \Ex(\min(\,|\residue I \cap \residue x \residue S|, |\residue I \cap \residue y \residue S|\,))
\intertext{Recall that $\residue I$ is an interval and note that $\residue y \residue S$ is exactly of the form $\residue y \residue x^{-1} \residue X$ for coprime integers~$x, y$ and an interval $\residue X$ of size $x$. Therefore, it follows by \autoref{lem:discrepancy-inverses} that $|\residue I \cap \residue y \residue S| \leq \frac{|\residue I| \mult |\residue y \residue S|}p + 2 = \frac{x |\residue I|}p + 2$:}
    \leq{} &\frac1x \Ex\left(\min\left(|\residue I \cap \residue x \residue S|, \frac{x |\residue I|}p + 2\right)\right)
\intertext{Next, since both $\residue I$ and $\residue x \residue S$ are intervals, there are less than $|\residue x \residue S| + |\residue I| = x + |\residue I|$ choices of $\residue\sigma$ such that $\residue I \cap \residue x \residue S$ is non-empty. We conclude that:}
    \leq{} &\frac{x + |\residue I|}{px} \left(\frac{|\residue I| \, x}p + 2\right) \\
    \leq{} &\frac{|\residue I|^2}{p^2} + \frac{3|\residue I|}{px} + \frac2p,
\end{align*}
recalling that $x = H_p(\residue x \residue y^{-1}) < \sqrt p$ (by \autoref{lem:heights-equivalence}). The claim follows.
\end{proof}

We are finally ready to prove \autoref{thm:threewise-independence} and \autoref{cor:overfull-buckets}. The proofs are analogous to \cite[Theorems~5 and~6]{Knudsen16}.

\begin{proof}[Proof of \autoref{thm:threewise-independence}]
Fix $y, z \in X$. We will later unfix $y$ and $z$ and consider them to be random variables. Let $I = \{ i \in [p] : i \bmod m = a \}$ and define $J$ similarly with $b$ in place of $a$; clearly $\residue I$ and $\residue J$ are arithmetic progressions in $\Int_p$. Let $\residue h(\residue x) = \residue\sigma \residue x + \residue\tau$ be a random linear function on $\Int_p$. Then:
\begin{align*}
    &\Pr(h(y) = h(z) = b \mid h(x) = a) \\
    ={} &\Pr((\residue h(\residue y), \residue h(\residue z)) \in \residue J^2 \mid \residue h(\residue x) \in \residue I) \\
    ={} &\frac1{|\residue I|} \sum_{\residue u \in \residue I} \Pr((\residue h(\residue y), \residue h(\residue z)) \in \residue J^2 \mid \residue h(\residue x) = \residue u).
\end{align*}
The last equality is by conditioning on $\residue h(\residue x)$ taking some fixed value $\residue u$. As $\residue h(\residue x)$ and $\residue h(\residue y) - \residue h(\residue x) = \residue\sigma (\residue y - \residue x)$ are independent, we can omit the condition and apply \autoref{lem:progression-independence}:
\begin{align*}
    &\Pr((\residue h(\residue y), \residue h(\residue z)) \in \residue J^2 \mid \residue h(\residue x) = \residue u) \\
    ={} &\Pr((\residue\sigma(\residue y - \residue x), \residue\sigma(\residue z - \residue x)) \in (\residue J - \residue u)^2) \\
    \leq{} & \frac{|\residue J|^2}{p^2} + \Order\!\left(\frac1{H_p(\frac{\residue y - \residue x}{\residue z - \residue x})} \mult \frac{|\residue J|}p + \frac1p\right)
\intertext{By the definition of $\residue J$ we have $|\residue J| \leq \ceil{\frac pm} \leq \frac pm + 1$ and thus:}
    \leq{} & \frac1{m^2} + \Order\!\left(\frac1{H_p(\frac{\residue y - \residue x}{\residue z - \residue x})} \mult \frac1m + \frac1p\right)\!.
\end{align*}
Now we unfix $y, z$ and consider $y, z \in X$ to be chosen uniformly at random. By averaging over the previous inequalities we get:
\begin{align*}
    &\Pr(h(y) = h(z) = b \mid h(x) = a) \\
    ={} &\frac1{m^2} + \Order\!\left(\frac1{k^2 m} \sum_{y, z \in X} \frac1{H_p(\frac{\residue y - \residue x}{\residue z - \residue x})} + \frac1p\right).
\end{align*}
As $y - x$ and $z - x$ are nonzero integers of magnitude at most $U < \sqrt{p / 2}$, we can apply \autoref{lem:heights-equivalence} to replace $H_p$ by $H$. The remaining sum can be bounded using \autoref{lem:heights-small-universe}:
\begin{equation*}
    \sum_{y, z \in X} \frac1{H_p(\frac{\residue y - \residue x}{\residue z - \residue x})} = \sum_{y, z \in X} \frac1{H(\frac{y - x}{z - x})} = \Order(U \log U),
\end{equation*}
and the claim follows. (The $+\frac1p$ term can be omitted as we are assuming that $p = \Omega(U^2)$.)
\end{proof}

\begin{proof}[Proof of \autoref{cor:overfull-buckets}]
Fix buckets $a, b \in [m]$, and let $F$ denote the number of keys in $X$ hashed to~$b$. By the pairwise independence of $h$ (see \autoref{lem:linear-hashing-basics}), we have that
\begin{equation*}
    \Ex(F \mid h(x) = a) = \Ex(F) = \frac km \pm \Theta\!\left(\frac kp\right) = \frac km \pm \Theta(1).
\end{equation*}
In particular, since $p > m^2$ it holds that $\Ex(F) \geq k/m - \Order(k/p) \geq \Omega(k/m)$. By \autoref{thm:threewise-independence}, we additionally have
\begin{equation*}
    \Ex(F^2 \mid h(x) = a) = \frac{k^2}{m^2} + \Order\!\left(\frac{U \log U}m\right)\!.
\end{equation*}
It follows that
\begin{equation*}
    \Var(F \mid h(x) = a) = \Order\!\left(\frac{U \log U}m\right)\!,
\end{equation*}
and finally, by an application of Chebyshev's inequality we have
\begin{equation*}
    \Pr(|F - \Ex(F)| \geq \lambda \sqrt{\Ex(F)} \mid h(x) = a) \leq \frac{\Var(F \mid h(x) = a)}{\lambda ^2 \Ex(F)} = \Order\!\left(\frac{U \log U}{\lambda^2 k}\right)\!,
\end{equation*}
for all $\lambda > 0$.
\end{proof}

\subsection{An Almost-Matching Lower Bound Against \autoref{thm:knudsen-threewise-independence}} \label{sec:hashing-lower-bound}
In this section we prove the following statement which shows that \autoref{thm:knudsen-threewise-independence} (Theorem~5 in~\cite{Knudsen16}) is almost optimal in the case where $U$ is polynomial in $k$.

\begin{theorem}[\autoref{thm:knudsen-threewise-independence} is Almost Optimal] \label{thm:threewise-independence-lower-bound}
Let $k$ and $U$ be arbitrary parameters with $U \geq k^{1+\epsilon}$ for some constant $\epsilon > 0$, and let $h$ be a random linear hash function with arbitrary parameters $p > U$ and $m \leq U$. Then there exists a set $X \subseteq [U]$ of $k$ keys, a fixed key $x \not\in X$ and buckets $a, b \in [m]$ such that
\begin{equation*}
    \Pr(h(y) = h(z) = b \mid h(x) = a) \geq \frac1{mk} \exp\!\left(\Omega\!\left(\sqrt{\min\left(\tfrac{\log k}{\log\log k}, \tfrac{\log U}{\log^2 \log U}\right)}\right)\!\right)\!,
\end{equation*}
where $y, z \in X$ are uniformly random.
\end{theorem}

We first describe the construction of $X$. By the Prime Number Theorem, for any $n \in \Nat$, there are $n$ primes $p_0, \dots, p_{n-1}$ in the range $[n \log n, C n \log n]$ for some absolute constant $C > 1$. Let
\begin{equation*}
    n = \min\left(\frac{\epsilon \log U}{(1 + \epsilon) C \log\log U},\, \log k\right)\!,
\end{equation*}
and define
\begin{equation*}
    X' = \left\{ s \prod_{i \in I} p_i : I \subseteq [n],\, |I| = \frac n2,\, 1 \leq s \leq S \right\}\!,
\end{equation*}
where $1 \leq S \leq k$ is chosen in such a way that $k/2 \leq |X'| \leq k$. There exists indeed such a value of $S$, since $S \leq |X'| \leq S \mult 2^n$ and $2^n \leq k$. We then construct $X \supseteq X'$ by adding arbitrary (small) elements to $X$ until $|X| = k$. One can check that $X \subseteq [U]$ as the largest number in $X'$ has magnitude less than
\begin{equation*}
    S (C n \log n)^n \leq k (\log U)^{\frac{\epsilon \log U}{(1+\epsilon) \log\log U}} = k U^{\frac\epsilon{1+\epsilon}} \leq U^{\frac1{1+\epsilon}} U^{\frac\epsilon{1+\epsilon}} = U.
\end{equation*}
The first step towards proving that $X$ is an extreme instance is to give the following lower bound:

\begin{lemma} \label{lem:heights-lower-bound}
It holds that
\begin{equation*}
    \sum_{x, y \in X} \frac1{H(\frac xy)} = k \exp\!\left(\Omega\!\left(\sqrt{\min\left(\tfrac{\log k}{\log\log k}, \tfrac{\log U}{\log^2 \log U}\right)}\right)\!\right)\!.
\end{equation*}
\end{lemma}
\begin{proof}
We only need a lower bound, so we will ignore all elements in $X \setminus X'$. Fix any element $x = s \prod_{i \in I} p_i \in X'$; we prove a lower bound against $\sum_{y \in X'} 1/H(\frac xy)$. We call an element~$y$ \emph{good} if it has the form $y = s \prod_{i \in J} p_i$, where both $x$ and $y$ have the same factor $s$ and the symmetric difference of $I$ and $J$ has size exactly $2r$ (i.e., $|I \setminus J| = |J \setminus I| = r$) for some parameter~$r$ to be specified soon. In the fraction $\frac xy$ only the factors $s$ and $\prod_{i \in I \intersect J} p_i$ cancel and therefore $H(\frac xy) = \Theta(n \log n)^r$. Moreover, the number of good elements $y$ is exactly \raisebox{0pt}[0pt][0pt]{$\binom{n/2}r{}^2$}, and thus
\begin{equation*}
    \sum_{y \in X'} \frac1{H(\frac xy)} \geq \frac{\binom{n/2}{r}^2}{\Order(n \log n)^r} \geq \frac{(\frac n{2r})^{2r}}{\Order(n \log n)^r} = \Omega\!\left(\frac{n}{r^2 \log n}\right)^r\!.
\end{equation*}
Choosing $r = \Theta(\sqrt{n / \log n})$ yields
\begin{equation*}
    \sum_{y \in X'} \frac1{H(\frac xy)} \geq \exp(\Omega(r)) = \exp\!\left(\Omega\!\left(\sqrt{\min\left(\tfrac{\log k}{\log\log k}, \tfrac{\log U}{\log^2 \log U}\right)}\right)\!\right)\!.
\end{equation*}
Recall that $x \in X'$ was arbitrary and thus the claim follows by summing over all $x \in X'$.
\end{proof}

\begin{proof}[Proof of \autoref{thm:threewise-independence-lower-bound}]
Let $X$ be as before, choose $x = 0$ and choose $a = b = 0$. For now we also fix $y, z \in X$ but we will later in the proof unfix $y, z$ and treat them as random variables. The first steps are quite similar to the proof of \autoref{thm:threewise-independence}. Let $I = \{ i \in [p] : i \bmod m = 0 \}$. Then $\residue I$ is an arithmetic progression in $\Int_p$. For $\residue h(\residue x) = \residue\sigma \residue x + \residue\tau$ a random linear function on $\Int_p$, we obtain:
\begin{align*}
    &\Pr(h(y) = h(z) = b \mid h(x) = a) \\
    ={} &\Pr((\residue h(\residue y), \residue h(\residue z)) \in \residue I^2 \mid \residue h(\residue x) \in \residue I) \\
    ={} &\frac1{|\residue I|} \sum_{\residue u \in \residue I} \Pr((\residue h(y), \residue h(\residue z)) \in \residue I^2 \mid \residue h(\residue x) = \residue u).
\end{align*}
We continue to bound every term in the sum from below, so fix some value $\residue u \in \residue I$. As $\residue h(\residue x) = \residue\tau$ and $\residue h(\residue y) - \residue h(\residue x) = \residue\sigma \residue y$ are independent, we can omit the condition and it suffices to bound
\begin{align*}
    &\Pr((\residue h(\residue y), \residue h(\residue z)) \in \residue I^2 \mid \residue h(\residue x) = \residue u) \\
    ={} &\Pr((\residue\sigma \residue y, \residue\sigma \residue z) \in (\residue I - \residue u)^2)
\intertext{Let $T = \floor{\frac p{2m}}$ and observe that either $\residue J = \{ \residue i \residue m : i \in [T] \}$ or $\{ -\residue i \residue m : i \in [T] \}$ is contained in $\residue I - \residue u$. In both cases we may replace $\residue I - \residue u$ by $\residue J$ (in the latter case we replace also $\residue\sigma$ by $-\residue\sigma$ which does not change the probability):}
    \geq{} &\Pr((\residue\sigma \residue y, \residue\sigma \residue z) \in \residue J^2) \\
    ={} &\frac{|\residue y^{-1} \residue J \cap \residue z^{-1} \residue J|}p \\
    ={} &\frac1p \sum_{i, j \in [T]} \left[ \residue i \residue y^{-1} \residue m = \residue j \residue z^{-1} \residue m \right] \\
    ={} &\frac1p \sum_{i, j \in [T]} \left[ \residue i \residue j^{-1} = \residue y \residue z^{-1} \right]
\intertext{We claim that there are at least $\Omega(T / H(\frac yz))$ solutions $i, j \in [T]$ to the equation $\residue i \residue j^{-1} = \residue y \residue z^{-1}$. Indeed, consider the reduced fraction \raisebox{0pt}[0pt][0pt]{$\frac{y'}{z'} = \frac yz$} (i.e., $y'$ and $z'$ are coprime and $0 < y', z' \leq H(\frac yz)$ by definition). For any $t < T / H(\frac yz)$ we may pick $i = t y'$ and $j = t z'$. On the one hand we have that~$i$ and~$j$ have the correct size since $i = t y' < (T / H(\frac yz)) \mult H(\frac yz) = T$ and on the other hand the equation is satisfied by \raisebox{0pt}[0pt][0pt]{$\residue i \residue j^{-1} = \residue t \residue{y'} \residue t^{-1} \residue{z'}{}^{-1} = \residue y \residue z^{-1}$}. Hence:}
    \geq{} &\Omega\left(\frac T{p H(\frac yz)}\right) \\
    ={} &\Omega\left(\frac1{m H(\frac yz)}\right).
\end{align*}
We now unfix $y, z$ and consider them as random variables. By averaging over the previous inequality and by applying \autoref{lem:heights-lower-bound} we finally obtain:
\begin{align*}
    &\Pr(h(y) = h(z) = b \mid h(x) = a) \\
    \geq{} &\Omega\left(\frac1{mk^2} \sum_{y, z \in X} \frac1{H(\frac yz)}\right) \\
    \geq{} &\frac1{mk} \exp\!\left(\Omega\!\left(\sqrt{\min\left(\tfrac{\log k}{\log\log k}, \tfrac{\log U}{\log^2 \log U}\right)}\right)\!\right)\!. \qedhere
\end{align*}
\end{proof}

\bibliographystyle{plain}
\bibliography{refs}

\appendix

\section{Computations with Transposed Vandermonde Matrices} \label{sec:vandermonde}
The goal of this section is prove the following theorem:

\thmvandermonde*{}

Throughout, let $V = V(a)$ be as in the statement and let $W = W(a) = V(a)^T$ denote its transpose, i.e.\ a Vandermonde matrix. The proof of \autoref{thm:vandermonde} is by the so-called \emph{transposition principle}: First, classic algorithms show that $W x$ and $W^{-1} x$ can be computed by efficient arithmetic circuits (\autoref{lem:poly-multipoint}). Second, whenever $A x$ can be computed efficiently by an arithmetic circuit, then also $A^T x$ can be computed similarly efficiently (\autoref{lem:transposition-principle}).

See~\cite{Rudra20,vonzurGathenG13} for a definition of \emph{arithmetic circuits}. For our purposes it suffices to only consider addition, subtraction and multiplication gates.

\begin{lemma}[Polynomial Evaluation and Interpolation] \label{lem:poly-multipoint}
There exists an algorithm $\mathcal A$ which, given $a \in F^n$ with pairwise distinct entries $a_i$, computes arithmetic circuits $C$ and $D$ such that:
\begin{itemize}[noitemsep]
\item $\mathcal A$ runs in time $\Order(n \log^2 n)$,
\item $\mathcal A$ uses at most $\Order(n \log^2 n)$ ring operations and at most 1 division,
\item On input $x \in F^n$, the circuit $C$ computes $C(x) = W(a) x$ and $D$ computes $D(x) = W(a)^{-1} x$.
\end{itemize}
\end{lemma}
\begin{proof}[Proof sketch]
For $C$, note that evaluating $C(x) = W(a) x$ is exactly the problem of evaluating the polynomial $\sum_i x_i X^i \in F[X]$ at the points $a_1, \dots, a_n$. Hence, we can use the classical $\Order(n \log^2 n)$-time algorithm for \emph{polynomial multi-point evaluation}, see e.g.~\cite[Algorithm~10.7]{vonzurGathenG13}. This algorithm can be interpreted to compute the arithmetic circuit $C$ rather than computing the multi-point evaluation directly. As this algorithm works over rings, it does not use any divisions.

For $D$, observe that computing $D(x) = W(a)^{-1} x$ corresponds to \emph{polynomial interpolation}, which again has a classical $\Order(n \log^2 n)$-time algorithm~\cite[Algorithm~10.11]{vonzurGathenG13}. This time however, we have to pay attention to the number of divisions performed in the process. Note that~\cite[Algorithm~10.11]{vonzurGathenG13} only computes divisions in the second step, all of which can be bulked together by \autoref{lem:bulk-division}, and moreover all inputs to these divisions only depend on $a$ and can thus be performed by the algorithm $\mathcal A$, rather than by the arithmetic circuit $D$. This proves the claim.
\end{proof}

Next, we need the following lemma which can be proven in several ways, for instance via the Baur-Strassen Theorem; see~\cite[Theorem 4.3.1]{Rudra20}.

\begin{lemma}[Transposition Principle] \label{lem:transposition-principle}
If for some matrix $A \in F^{n \times n}$, the function $x \mapsto A x$ is computed by an arithmetic circuit $C$ of size $s$, then the function $x \mapsto A^T x$ is computed by an arithmetic circuit $C'$ of size $\Order(s + n)$. Moreover, one can compute $C'$ from $C$ in time $\Order(s + n)$.
\end{lemma}

\begin{proof}[Proof of \autoref{thm:vandermonde}]
First, we run the algorithm $\mathcal A$ from \autoref{lem:poly-multipoint} to compute arithmetic circuits~$C$ and~$D$ of size $s = \Order(n \log^2 n)$. Recall that $\mathcal A$ uses $\Order(n \log^2 n)$ ring operations plus a single division. The circuits compute $C(x) = W(a) x$ and $D(x) = W(a)^{-1} x$, respectively. Second, use the transposition principle (\autoref{lem:transposition-principle}) to compute circuits $C'$ and $D'$ with $C'(x) = W(a)^T x = V(a) x$ and $D'(x) = (W(a)^{-1})^T x = V(a)^{-1} x$. By \autoref{lem:transposition-principle}, $C'$ and $D'$ also have size $\Order(s + n) = \Order(n \log^2 n)$ and we can compute both in the same running time. Finally, we evaluate $C'$ and $D'$ at~$x$, which again takes time $\Order(n \log^2 n)$ and uses only ring operations.
\end{proof}

\end{document}